\documentclass[a4paper, twocolumn, superscriptaddress,
               amsmath, amssymb, accepted=2026-11-17]{quantumarticle}
\pdfoutput=1 
\usepackage{graphicx}
\usepackage{hyperref}
\usepackage[numbers]{natbib} 
\usepackage[linesnumbered, ruled, vlined, titlenotnumbered]{algorithm2e}
\usepackage{amsmath, amsfonts, amssymb, amsthm, color, booktabs}
\usepackage{array}
\usepackage[capitalise]{cleveref}
\crefname{algocf}{alg.}{algs.}
\Crefname{algocf}{Algorithm}{Algorithms}

\newtheorem{theorem}{Theorem}
\newtheorem{definition}{Definition}
\newtheorem{proposition}{Proposition}

\newtheorem{lemma}{Lemma}

\newcommand{\ket}[1]{|#1\rangle}
\newcommand{\tree}{{\cal T}}
\DeclareMathOperator{\CliNR}{CliNR}
\newcommand{\N}{{\mathbb{N}}}

\newcommand{\plog}{{p_{\log}}}
\DeclareMathOperator{\restart}{res}
\DeclareMathOperator{\ovspace}{Q}
\DeclareMathOperator{\ovtime}{G}

\newcommand{\qubitoverhead}{{\omega_{\ovspace}}}
\newcommand{\gateoverhead}{{\omega_{\ovtime}}}

\newcommand{\gp}{g_p}
\newcommand{\pres}{{p_\text{res}}}
\newcommand{\av}{A_V}
\newcommand{\ap}{A_P}
\newcommand{\bv}{B_V}
\newcommand{\ai}{A_I}
\newcommand{\vp}{\vec{P}}
\newcommand{\ts}{\hat{s}}
\newcommand{\ie}{{\em i.e.~}}
\newcommand{\rsp}{\text{RSP}}
\newcommand{\rsv}{\text{RSV}}
\newcommand{\rsi}{\text{RSI}}

\begin{document}

\title{Recursive Clifford noise reduction}
\author{Aharon Brodutch}
\affiliation{IonQ Inc.}

\author{Gregory Baimetov }
\affiliation{IonQ Inc.}
\affiliation{Department of Mathematics, University of Washington, Seattle, WA 98195-4350}

\author{Edwin Tham}
\affiliation{IonQ Inc.}

\author{Nicolas Delfosse}
\affiliation{IonQ Inc.}


\begin{abstract}
Clifford noise reduction  $\CliNR$ is a partial error correction scheme that reduces the logical error rate of Clifford circuits at the cost of a modest qubit  and gate  overhead. The $\CliNR$ implementation of an  $n$-qubit Clifford circuit of size $s$ achieves a vanishing logical error rate if $snp^2\rightarrow 0$ where $p$ is the physical error rate.
Here, we propose a recursive version of $\CliNR$ that can reduce errors on larger circuits with a relatively small gate overhead. 
When $np \rightarrow 0$, the logical error rate can be vanishingly small. When $sp >2$ this can be achieved with  $\left(2\left\lceil \log(sp)\right\rceil+1\right)n+1$ qubits and at most $24 s \left\lceil(sp)^{3.6}\right\rceil $ gates. For $sp\le 2$ we recover the original $\CliNR$.
Using numerical simulations, we show that the recursive method can offer an advantage in a realistic near-term parameter regime. When circuit sizes are large enough, recursive $\CliNR$ can reach a lower logical error rate than the original $\CliNR$ with the same gate overhead.  The results offer promise for reducing logical errors in large Clifford circuits with relatively small overheads. 
\end{abstract}

\maketitle

\section{Introduction}

Two major obstacles for reaching utility scale quantum computation are increasing the number of qubits and reducing error rates. Quantum error correction (QEC) offers a way to reduce the logical error rates at the cost of increasing physical resources in a scalable way \cite{aharonov1997fault, aliferis2005quantum}. Current quantum computers do not support enough qubits for full-fledged fault-tolerant QEC. Moreover, the overheads for  implementing  fault-tolerant logical operations are high, both in qubit and gate counts, notwithstanding huge improvements over the last few years~\cite{Litinski_2019, bravyi2024high, gidney2025factor2048bitrsa, ibm}. At the same time, physical error rates and raw qubit counts are already sufficient for showing the advantage of quantum computers~\cite{arute2019quantum,DeCross_2025, Morvan_2024}. These results provide some evidence that fault-tolerant QEC might not be necessary to show some initial utility scale computation, especially if methods for error mitigation and partial error correction can be used to reduce the impact of errors \cite{myths}.

Error-mitigation (EM) refers to a collection of techniques aimed at reducing the impact of errors on the output of a quantum computer.
Typically, EM techniques re-shape noise ({\it e.g.} twirling \cite{twirl}), post-select shots identified as erroneous (e.g. flag error-detection \cite{debroy2020extended, ibmflags}), post-process noisy data ({\it e.g.} filtering \cite{Chen_2024}), or a combination thereof \cite{neutrino}. 
Of particular interest are coherent Pauli checks (CPC) \cite{roffe2018protecting,debroy2020extended,ibmflags} that require no explicit encoding, and are readily applied to application circuits with few ancilla qubit overheads. The main disadvantages of CPC are that it is limited to reducing errors in Clifford circuits and that any single detected error requires the shot to be discarded and the computation restarted anew. As a result of the latter, the restart rate approaches unity, as circuit size grows. Clifford Noise Reduction $\CliNR$ \cite{delfosse2024low} deals with this problem by building Clifford subcircuits offline and injecting them into the main circuit when they pass verification. The result is a low-overhead scheme that can be used to reduce logical error rates in Clifford subcircuits. The method is however limited to subcircuits of size at most of the order of $1/(np^2)$ where $p$ is the physical error rate and $n$ is the number of qubits. Here we present a generalization of $\CliNR$ which we call \emph{Recursive $\CliNR$}. It provides low-overhead error reduction for sub-circuits of any size $s$, as long as $np$ is sufficiently small. As in \cite{delfosse2024low}, we provide an analytical proof that this is asymptotically true for a simple noise model with no idle noise and provide numerical evidence for  more complex noise models. 

Recursive $\CliNR$, recursively applies $\CliNR$ to a sequence of $\CliNR$ sub-circuits. The main advantage is that the injection circuits, which don't get verified in the standard $\CliNR$, will now get verified at the next level of the recursion. However, new unverified injection steps are introduced at the next level. These get verified at the next level {\it etc.}, with the aim of ultimately having a small number of unverified injection steps at the top level. For our main result, we show that the number of unverified injection steps can be reduced to a constant, independent of $s$ and $n$, while keeping the gate overhead polynomial in $sp$.  This is in contrast to the original $\CliNR$ scheme where the number of unverified injection steps must grow like $\sqrt{s/n}$ in order to keep the error low and avoid an exponential overhead. 

The number of $\CliNR$ subcircuits (each with its injection step) at each level of the recursion is determined by a tree which we call the \emph{$\CliNR$ tree}. The right choice of parameters for this tree can lead to better performance than standard $\CliNR$ - lower logical error rates at the same circuit size. 

Our main contributions are as follows: We introduce \emph{recursive $\CliNR$}, a strict generalization of $\CliNR$, and the $\CliNR$ tree which is used as part of the implementation (\cref{algorithm:generalized_CliNR_circuit}).  We introduce the \emph{uniformly bounded} implementation  (\cref{algo:bounded}) and use it to prove that recursive $\CliNR$ can be used to achieve vanishing logical error rates with a polynomial gate overhead when $np\rightarrow 0$ (\cref{theorem:main}).  This result is a bound derived for a simple uniform noise model with no idle noise. For a more realistic noise model, and tighter results, we provide numerical results  in \cref{sec:numerical} based on Monte Carlo simulations and a Markov model (derived in \cref{sec:optimizing}). 

The rest of this paper is organized as follows. In the next section we review some background material including the noise model and the original $\CliNR$ scheme. We identify the main source of errors that cannot be detected and the limitation of the scheme.  In \cref{sec:recursive} we introduce \emph{Recursive $\CliNR$} and the \emph{$\CliNR$ tree} used in its construction. An algorithm for constructing the corresponding circuit is provided in \cref{algorithm:generalized_CliNR_circuit}. In \cref{sec:boundedtree} we introduce a specific implementation of Recursive $\CliNR$ called the \emph{uniformly bounded} implementation. It   allows us to make our main statement (\cref{theorem:main}): Recursive $\CliNR$ can be used to reach a vanishing logical error rate with a gate overhead that's polynomial in $sp$, as long as $np\rightarrow 0$. We do this using a sequence of technical statements which we prove in \cref{sec:technical}.  In \cref{sec:numerical} we present numerical results that show the advantage of Recursive $\CliNR$ at $n=70$ and $n=400$. In the conclusion section (\cref{sec:conclusions}) we discuss potential improvements to the implementation.

\section{Background}

In this section, we review the standard circuit-level noise model and the original CliNR scheme~\cite{delfosse2024low}.

\subsection{Noise model} \label{sec:noise_model}

We consider quantum circuits containing single-qubit preparations, unitary gates and single-qubit measurements.
An important class of quantum circuits is the set of $n$-qubit {\em Clifford circuits}. These circuits act on $n$-qubit input states in the same way as a Clifford unitary $U$ (potentially after discarding ancillas). A Clifford unitary preserves   Pauli operators by conjugation, \ie for all Pauli operators $P$, the operation $UPU^{\dagger}$ is also a Pauli operator. 

We consider the standard circuit-level noise model where each circuit operation is followed by a random Pauli error with rate $p$. The Pauli error is selected uniformly in the set of non-trivial Pauli errors acting on the support of the operation.
Measurement noise is represented by a flip of the measurement outcome with probability $p$.
Errors corresponding to different operations are assumed to be independent. In our simulation we also include idle noise and use different error rates for different options (see \cref{sec:numerical}).  

The circuit-level noise model introduces errors in a quantum circuit, which translate into errors on the output state of the circuit, that we call the {\em output error}.

An {\em implementation} of a Clifford circuit $C$ is a circuit $C'$ that produces the same output state as $C$ for any input state in the absence of error.
The circuit-level noise induces a probability distribution over the set of output errors of an implementation.
The {\em logical error rate} of an implementation is defined to be the probability that the output error has a non-trivial effect on the output of the circuit. We say that an implementation is \emph{ideal} when the logical error rate is 0, otherwise we say that the implementation is noisy. 

The most naive implementation of a circuit is its {\em direct implementation} $C'=C$.
In general, an implementation $C'$ of a Clifford circuit $C$ uses additional qubits and gates. 
The extra qubits are traced out before the end of the circuit.
The {\em qubit overhead}, $\qubitoverhead$,  is defined as the number of qubits of $C'$ divided by the number of qubits in $C$.
The {\em gate overhead}, $\gateoverhead$,  is the expected number of gates in $C'$ divided by the number of gates in $C$.
A gate may be executed multiple times because of a repeat until success loop for instance. 
Then, it is counted as many times as it is repeated in the gate overhead.

Throughout this work we divide the implementation of a Clifford circuit $C$ into blocks. If a fault during block $B$ has a non-trivial effect on the output of $B$ then we say that it is a \emph{logical error in $B$}. A logical error in $B$ does not necessarily imply a logical error in the implementation of  $C$.

\subsection{Review of CliNR}\label{sec:clinr} 

$\CliNR$ was introduced in \cite{delfosse2024low} as a protocol that can reduce the logical error rate in the implementation of a Clifford circuit. 
Consider a unitary $n$-qubit Clifford circuit $C$, called the \emph{input circuit}, with size $s$ ({\it i.e.} $C$ has $s$ unitary gates). Given an integer $r$, the $\CliNR$ implementation, denoted $\CliNR_{1,r}(C)$, is defined as the following three-step process (see \cref{fig:clinr1}) : 
\begin{enumerate}
    \item {\bf Resource state preparation (RSP)}. A $2n$-qubit resource state (the Choi isomorphic state of $C$), $\ket{\varphi}$ is prepared as follows. First $n$ Bell pairs are prepared (using at most $A_{P} n$ noisy gates). 
    Next, the circuit $C$ is implemented on $n$ of these qubits, one from each pair. The implementation of $C$ requires $s$ noisy gates. 
    
    The resource state $\ket{\varphi}$ generated by an ideal implementation of RSP is  stabilized by a stabilizer group $\mathcal{S}(\varphi)$. An error during the noisy implementation of RSP produces a state $E\ket{\varphi}$. If $E$ does not commute with $\mathcal{S}(\varphi)$ then it is a logical error in RSP. 

    \item {\bf Resource state verification (RSV)}. The resource state is verified by measuring $r$ stabilizers using \emph{stabilizer measurements}.  Here we assume that the stabilizers are chosen randomly and uniformly from $\mathcal{S}(\varphi)$, but other choices can lead to better performance~\cite{tham2025optimized}. A logical error in RSP has probability $1/2$ to commute with a stabilizer chosen at random from $\mathcal{S}(\varphi)$. In an ideal implementation of the stabilizer measurement, a logical error that anti-commutes with the stabilizer is detected. If an error is detected we restart from RSP. In a noisy implementation an error during the stabilizer measurement can lead to a false detection or to a subsequent logical error. Each stabilizer measurement requires a single ancilla. We can define constants $A_V$ and $B_V$ such that the implementation of RSV requires at most $A_Vn + B_V$ noisy gates. An ideal implementation will reduce the logical error rate by a factor of $2^{r}$ (see \cref{lemma:checks}). 
    
    \item {\bf Resource state injection (RSI)}. If no errors are detected during RSV, the subcircuit $C$ is applied to an $n$-qubit input state by consuming the resource state $\ket{\varphi}$ and performing a measurement-outcome-dependent Pauli correction, as in standard gate teleportation \cite{gottesman1999quantum}. Note that the correction is a Pauli because $C$ is Clifford. As a result we can define a constant $A_I$ such that this process takes at most $A_I n$ noisy gates (independent of $C$).  The rails hosting the input state  (top 3 rails in \cref{fig:clinr1} (a) and (b)) are called the \emph{input rails}. 
    
\end{enumerate}
Here and throughout this work, $A_P,\; A_V, \; B_V $ and $A_I$ are constants as defined above. These depend on the specifics of the implementation, but are independent of $s$,  $p$, $n$ and $C$.  Note that we assume all-to-all connectivity. In an implementation with limited connectivity, the constants might depend on $n$. We refer to the three parts above as the RSP, RSV and RSI blocks {\it e.g.} RSI block refers to  the implementation of the RSI circuit in \cref{fig:clinr1} (b). A $\CliNR_1$ block refers to the implementation of the entire circuit in \cref{fig:clinr1} (b). We use $\CliNR_1$ as a general term for the $\CliNR$ protocol above, to distinguish it from Recursive $\CliNR$ defined in \cref{sec:recursive}.

The expected number of restarts in $\CliNR_1$ is exponential in $s$. To keep the overhead small, we break $C$ into $t$ subcircuits $C_i$ and then apply $\CliNR_{1,r}$ to each subcircuit. The logical error rate and restart rate in each $\CliNR_{1,r}(C_i)$ block are lower than in $\CliNR_{1,r}(C)$. In \cite{delfosse2024low} this technique (called $\CliNR_t$) was used to show that it is possible to achieve a vanishing logical error rate when $snp^2\rightarrow 0$. The gate overhead in this regime can be bounded by $\gateoverhead \le 2$. This can be contrasted with other techniques like coherent Pauli checks \cite{ibmflags} where the number of discarded shots is exponential in $s$.

One disadvantage of this technique is that errors during each RSI block cannot be detected. In the following, we tackle this issue with a recursive $\CliNR$ scheme. This scheme exploits the fact that $\CliNR_1$ blocks are Clifford circuits which can themselves be used as the input circuit for a larger $\CliNR_1$ block. 

\section{Recursive CliNR } \label{sec:recursive}

This section introduces the recursive generalization of $\CliNR$ which we call \emph{Recursive $\CliNR$}. We begin with the $\CliNR$ tree. 

An {\em ordered tree} is a rooted tree such that for any vertex $u$, the children of $u$ have a specific order $v_0, v_1,$ {\em etc}.
If a tree is ordered, we can label its vertices as 
$v_{\ell, i}$, where $\ell \in \N$ is the distance to the root, which we call the {\em level}, and $i$ is the index of the vertex among the vertices at level $\ell$.
For example, the vertices of a 4-regular tree are indexed as $v_{0, 0}, v_{1, 0}, v_{1, 1}, v_{1, 2}, v_{1, 3}, v_{2, 0}, v_{2, 1}, $ {\em etc}.
\cref{fig:boundedT} (a) shows an example of an ordered tree with labeled vertices. 

A {\em $\CliNR$ tree}, denoted $\tree({\bf s}, {\bf r})$ or simply $\tree$, is a finite ordered tree, equipped with two functions ${\bf s}$ and ${\bf r}$ defined over the vertex set of $\tree$.
These functions assign non-negative integers ${\bf s}(u)$ and ${\bf r}(u)$ to each vertex $u$ of the tree. 
We require that, for any vertex $u$, the sum of the values ${\bf s}(v)$ associated with its children $v$ is equal to ${\bf s}(u)$.

Let $\tree$ be a $\CliNR$ tree with depth $D$ and let $C$ be an $n$-qubit Clifford circuit with size ${\bf s}(v_{0,0})$.
For all levels $\ell=0,1,\dots, D$, we have
$
\sum_{i} {\bf s}(v_{\ell, i}) = {\bf s}(v_{0,0})
$
where $i$ runs over the indices of all the vertices at level $\ell$ of the tree.
It is possible to write  $C$ as the concatenation of subcircuits
$
C_{\ell, 0}, C_{\ell, 1}, \dots
$
where $C_{\ell, i}$ is a subcircuit of $C$ with size ${\bf s}(v_{\ell, i})$.
Moreover, for all $\ell \leq D-1$, the circuit $C_{\ell, i}$ is the concatenation of the subcircuits
$
C_{\ell+1, j}
$
where $v_{\ell+1, j}$ runs over the children of $v_{\ell, i}$ in $\tree$.

The {\em recursive $\CliNR$ implementation} of $C$, denoted $\CliNR(\tree, C)$, is the implementation obtained by applying the procedure described in \cref{algorithm:generalized_CliNR_circuit} to the circuit $C$.
\cref{fig:boundedT} (c) provides the circuit for the implementation of $\CliNR(\tree,C)$ where $\tree$ is the tree in \cref{fig:boundedT} (a).

\begin{algorithm}
\caption{Algorithm for constructing a recursive $\CliNR$ circuit}
\label{algorithm:generalized_CliNR_circuit}
\DontPrintSemicolon
\SetKwInOut{Input}{input}
\SetKwInOut{Output}{output} 
    \Input{An $n$-qubit Clifford circuit $C$ and a $\CliNR$ tree $\tree = \tree({\bf s}, {\bf r})$.}
    \Output{The implementation $\CliNR(\tree, C)$ of $C$.}
    Set $C' = C_{0,0}$.\;
    
    Add one ancilla to $C'$.
    
    \For{all level $\ell = 0,1,\dots, D-1$}
    {
        Add $2n$ ancilla qubits to $C'$.\;
        \For{all vertices $v_{\ell+1, j}$}
        {
            Let $r = {\bf r}(v_{\ell+1, j})$.\;
            Replace the sub-circuit $C_{\ell+1, j}$ inside $C'$ by $\CliNR_{1, r} (C_{\ell+1, j})$ using the $2n+1$ ancilla qubits added at level $\ell$.\;\label{algoline:replace}
        }
    }
    \Return{$C'$.}
\end{algorithm}

In \cref{algorithm:generalized_CliNR_circuit}, we do not use the value of $\bf r$ at the root. 
One could extend the construction by adding ${\bf r}(v_{0,0})$ stabilizer measurements at the end of the circuit.

The implementation $\CliNR(\tree, C)$ is a proper generalization of the $\CliNR$ implementations introduced in \cite{delfosse2024low}.
Indeed, we recover the implementation $\CliNR_{1, r}(C)$ by taking a tree with two vertices $v_{0,0}$ and $v_{1,0}$ such that ${\bf r}(v_{1,0}) = r$.
The implementation $\CliNR_{t, r}(C)$ is derived from a depth-one tree with a root and $t$ children $v_{1,j}$ for $j=0,1\dots, t-1$ with ${\bf r}(v_{1,j}) = r$.

\begin{figure}
    \centering
\hspace{-28pt}
\includegraphics[width=1\linewidth]{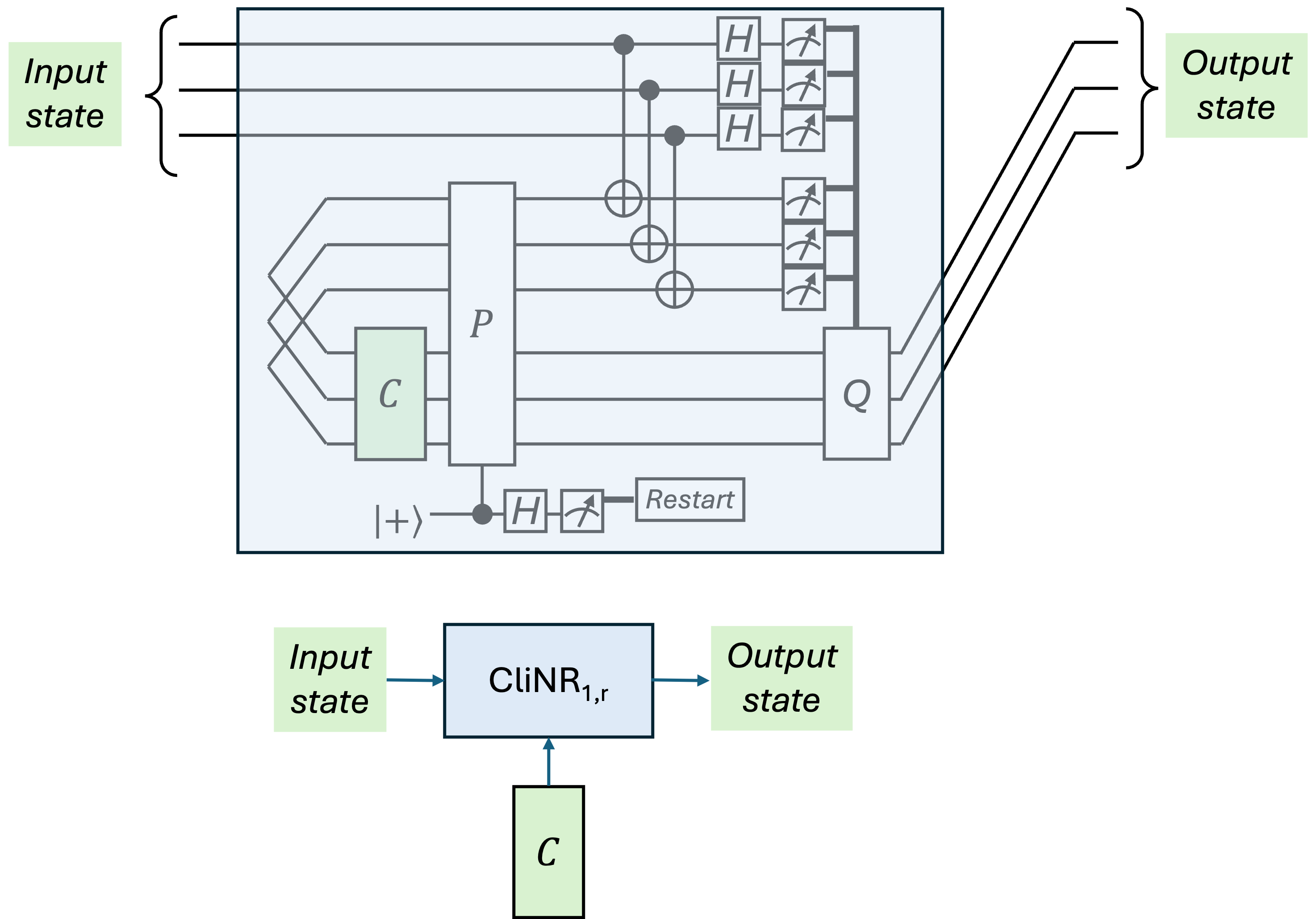} 
\vspace{-10pt}

(a)

\vspace{10pt}

\includegraphics[width=1\linewidth]{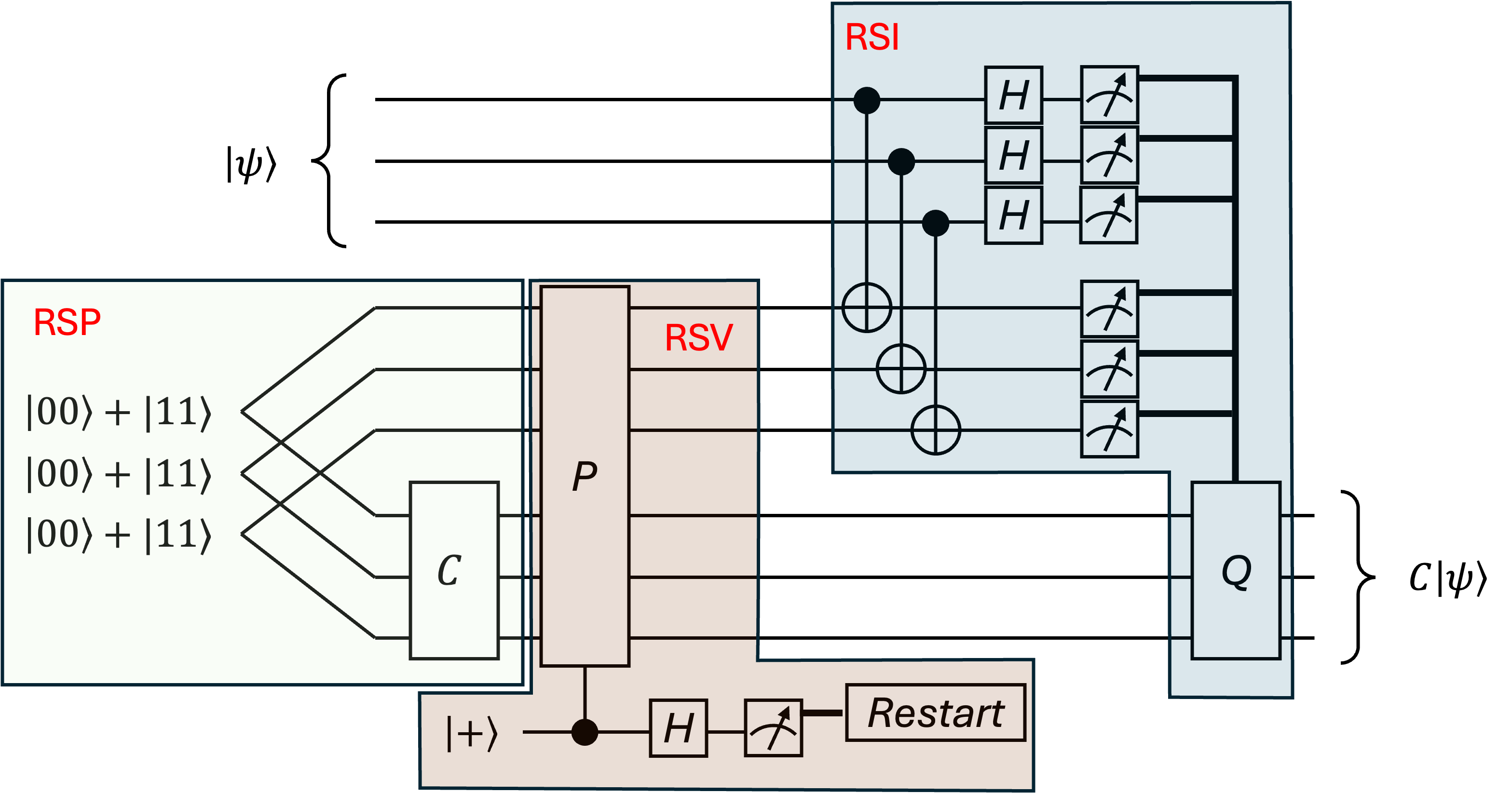}

    (b)

    \caption{A $\CliNR_1$ circuit. (a) Schematic of the main components for a $\CliNR$ circuit, \emph{input state}, \emph{output state}, and \emph{input circuit} $C$.  (b) The three blocks used to build a $\CliNR$ circuit: resource state preparation (RSP), resource state verification (RSV) and resource state injection (RSI). The number of noisy gates in each of these blocks is $\ap n + s$, $(\av n +\bv)r$ and $\ai n$ respectively. $\ap, \av, \bv, \ai$ are constants that depend on the implementation ({\it e.g.} number of noisy gates used to implement a stabilizer measurement). For simplicity only a single stabilizer measurement is shown in RSV, in general the RSV block will have multiple stabilizer measurements, all of these can use the same ancilla (bottom rail).  
    }
    \label{fig:clinr1}
\end{figure}

There are many possible choices for the stabilizers used to define  $\CliNR(\tree, C)$. 
The only requirement is that they belong to the stabilizer group of the resource state consumed in $\CliNR_{1, r}$ in \cref{algorithm:generalized_CliNR_circuit}. 
Here, we select random stabilizer generators following \cite{delfosse2024low}.
Optimization strategies to select these stabilizers more carefully were proposed recently in~\cite{tham2025optimized}. 

The following is a direct consequence of \cref{algorithm:generalized_CliNR_circuit}:
\begin{proposition}[qubit overhead for Recursive $\CliNR$] \label{lemma:qubit_overhead} The recursive implementation of an $n$-qubit Clifford circuit requires   $(2D+1)n + 1$ qubits.
\end{proposition} 

When $n \rightarrow +\infty$ the qubit overhead is $\qubitoverhead \rightarrow 2D+1$. 
For $D = 1$, we recover the  $3\times$ qubit overhead of the original CliNR scheme as expected.  In the conclusions (\cref{sec:conclusions}) we briefly discuss other implementations that may perform better at the cost of larger qubit overheads. Gate overheads and logical error rates are studied and discussed in what follows. 

\section{The uniformly bounded CliNR implementation }\label{sec:boundedtree}

Our main result below shows that when $np \rightarrow 0 $ it is possible to implement Recursive $\CliNR$ with a vanishing logical error rate. 
To assist with the proof, we define the \emph{uniformly bounded} implementation (for a formal definition  see \cref{def:bounded}). This implementation ensures that every iterative use of $\CliNR_{1,r}$ has an input circuit with a logical error rate bounded by $2/3$ and an output state with a logical error rate upper bounded by a function of order $np$ (see \cref{fig:boundedT} (b)).  The logical error rate  for the full implementation is therefore bounded by a function of order $t_1np$ where $t_1$ is the number of vertices at level 1. 
For sufficiently large $D$, we get $t_1=1$. In this case the logical error rate is bounded from above and below by functions  of order $np$. 

A schematic of the implementation is given in \cref{fig:boundedT}.  
At the bottom level ($\ell=D$) we choose subcircuit $C_{D,j}$ with size at  most $\bar{s}_D =  \frac{2}{3p}$. We then  
  choose  $r_{D,j}$ to bound the logical error rate from RSP to   the same order as  the logical error rate from RSV. This can be done by fixing all $r_{\ell,j}=R$ (defined in \cref{eq:R}).  The logical error rate for $\CliNR_1$ is then bounded by $Anp$ where $A$ is a constant (see  proof of \cref{lemma:logicalbounded}). 
At the next level up (and iteratively up to level 1) we take groups of $T$ vertices such that $T A np \le 2/3$ and feed them to the next vertex up ensuring that the logical error rate of each individual $\CliNR_1$ implementation in \cref{algorithm:generalized_CliNR_circuit} is bounded by $Anp$.

\begin{figure*}
    \centering

    \includegraphics[width=0.2\linewidth]{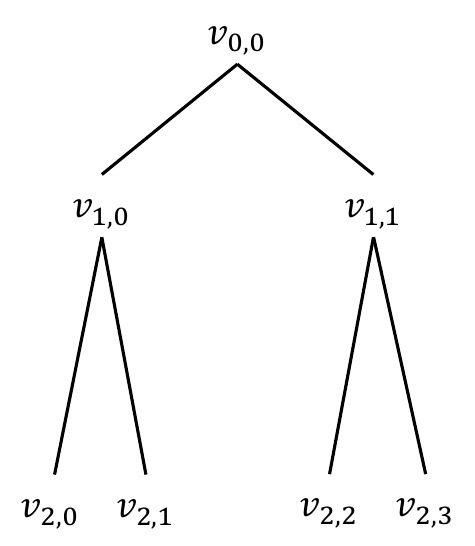}
    \hspace{1cm}
     \includegraphics[width=0.7\linewidth]{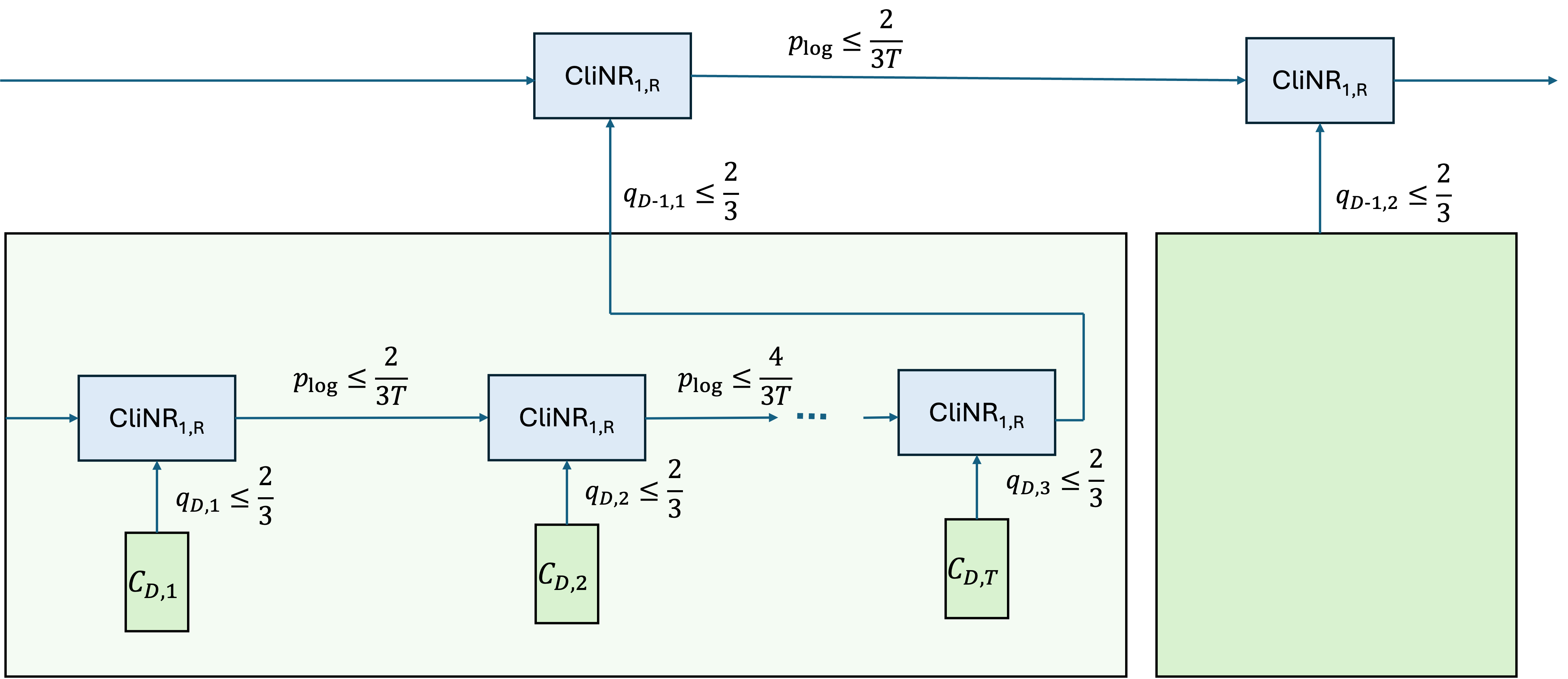}

    \hspace{-90pt} (a) \hspace{290pt} (b)

\vspace{10pt}

      \includegraphics[width=1\linewidth]{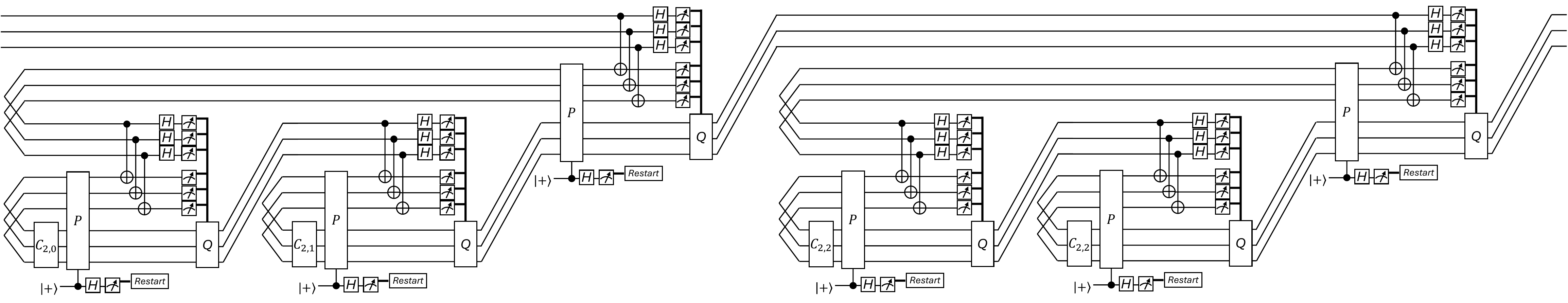}

      (c)
      
    \caption{Recursive $\CliNR$. (a) A depth 2 $\CliNR$ tree. Each vertex in the tree has two associated parameters ${\bf r}(v_{\ell,j})=r_{\ell,j}$ and ${\bf s}(v_{\ell,j})=s_{\ell,j}$ referring to the number of stabilizer measurements and size of the subcircuit $C_{\ell, j}$.  (b) Schematic for the uniformly bounded implementation of Recursive $\CliNR$ (see \cref{algo:bounded}).  Shown is only part of the implementation starting at the beginning of the circuit. Note that the \emph{input  state} for the first (bottom left) $\CliNR_1$ block is determined by the previous level.   The entire green block is the \emph{input circuit} to the level above. The specific choice of tree parameters (based on $s$) ensures that the logical error rate  at each input ($q_{\ell,j}$) and output ($\plog$) state/circuit is bounded. Each implementation of $\CliNR_{1,R}$ has an input circuit with logical error rate of at most $2/3$ and an output with logical error rate at most $2/(3T)$.  The total error is therefore upper bounded by $2t_1/(3T)$ where $t_1$ is the number of vertices at level 1. 
    (c) Circuit implementation of Recursive $\CliNR$ with the tree from (a) and $n=3$. For simplicity $r_{\ell,j}=1$ for all $(\ell,j) \ne (0,0)$. Note that the same ancilla can be used for all  stabilizer measurements. The total qubit count is $16$.}

    \label{fig:boundedT}
\end{figure*}

\subsection{Main result}

Here, we state the main result of this paper which proves that the recursive CliNR scheme introduced in this work can achieve a vanishing logical error rate for $n$-qubit Clifford circuits with arbitrary size as long as $np \rightarrow 0$.
For comparison, the standard CliNR scheme is limited to circuits of size $s = o(1/(np^2))$.

\begin{theorem} \label{theorem:main}
 Consider the family of $n$-qubit Clifford circuits $C$ of size $s$ with physical error rate $p$ and let  $D = \max\left\{1,\lceil \log(sp)\rceil\right\}$.   If $np\rightarrow 0$  there is a recursive $\CliNR$ implementation with a vanishing logical error rate, a  gate overhead  $\gateoverhead \le  24\left\lceil(sp)^{3.6}\right\rceil$, and a qubit overhead of $\qubitoverhead \le (2D+1)+1/n$. 
\end{theorem}

Here and throughout this paper, $\log$ denotes the base 2 logarithm and $\gateoverhead$, $\qubitoverhead$ are the $\CliNR$ gate and qubit overheads respectively. 
\begin{proof}
    If $sp\le 2 $, then standard $\CliNR$ is sufficient and we can use the bounds in \cite[Theorem 2]{delfosse2024low}.

    For $sp>2$ we use the uniformly bounded implementation (see \cref{sec:bounded}). We now prove that it produces the desired logical error rate and gate overhead. 
    
    Let $\alpha = 18(4\av n +2\bv )+3 \ai n$ (as in \cref{lemma:logicalbounded}).  
    From \cref{lemma:logicalbounded}, sufficient conditions for a vanishing logical error rate are  $$\left(\dfrac{1}{\alpha p}\right)^{D-1} \ge \dfrac{3sp}{2}$$ and  $(\alpha p)\rightarrow 0$. These are satisfied since $D\ge \log(sp)$ and $np\rightarrow 0$.  
    
    The bound for the gate overhead from  \cref{lemma:gateoverheadbound} is $\gateoverhead\le 2(12^D) < 24\left\lceil(sp)^{3.6}\right\rceil$. 
     The qubit overhead is given in \cref{lemma:qubit_overhead}.
\end{proof}

The condition $np\rightarrow 0$ cannot be avoided with CliNR since the logical error rate is always bounded from below by the logical error rate in RSI which has $\ai n$ noisy gates. 
We note that the bound on $\gateoverhead$ is loose (see \cref{sec:tightness} ). 

\subsection{Notation} \label{sec:notation}

We continue to use the notation from \cref{algorithm:generalized_CliNR_circuit} unless stated otherwise. We use  $r_{\ell,j}={\bf r}(v_{\ell,j})$ and $s_{\ell,j} = {\bf s}(v_{\ell,j})$. Note that we only consider pairs $\tree$, $C$ such that  $s_{0,0}=s$ is the size of $C$ and $r_{0,0}=0$. 

We use $\tree_{\ell,j}$ for a $\CliNR$ tree constructed by taking the subtree of $\tree$ consisting of the vertex $v_{\ell,j}$ and all its children, and setting $r_{0,0}=0$. 

We also use $\ts_{\ell,j}$ to denote the expected number of gates in the implementation of $\CliNR(\tree_{\ell,j},C_{\ell,j})$, accounting for the repetitions due to the restarts in all RSV of the tree. We use $\ts_\ell = \max_j ( \ts_{\ell,j})$. 

The constants $\ap, \,  \av, \, \bv, \, \ai$ are used to bound gate counts (as defined in  \cref{sec:clinr}). We define $m_{\ell,j} = \ap n+r_{\ell,j}(\av n+\bv)$, which we use to upper bound the size of a single RSP+RSV iteration. When there is no risk of confusion we use $m$ instead of $m_{\ell, j}$. We also use upper and lower bars to denote the max and min of a variable over all vertices in the same level {\it e.g.} $\bar{r}_{\ell} = \max_j (r_{\ell,j})$ and $\underline{r}_{\ell}= \min_j(r_{\ell,j}) $. 
We use the notation $\gp(x)=1-(1-p)^x$.

We use $\plog$ to denote the logical error rate in a specific context (usually for the  $\CliNR$ implementation). In cases where confusion is possible we use an argument {\it e.g.} $\plog(C) \le \gp(s)$ is the logical error rate in the direct implementation of  $C$. 

For each vertex $v_{\ell,j}$ of $\tree$, the {\em subtree error rate}, denoted $q_{\ell,j}$ is defined to be the logical error rate of the implementation $\CliNR(\tree_{\ell,j},C_{\ell,j})$ of $C_{\ell,j}$.
The logical error rate of the whole recursive implementation of $C$ is equal to $q_{0,0}$.
We use the notation $\bar{p}_{b\ell}$ for an upper bound on $\max_j q_{\ell,j}$.

An easy reference table for notation is provided in \cref{tab:notation}.

\begin{table*}[t]
    \centering
\begin{tabular}{|c|c|c|}
    \hline
         Notation&  Meaning& Notes\\\hline
 $\tree_{\ell,j}$& Subtree of $\tree$ &See \cref{sec:notation}\\\hline
         $C_{\ell,j}$ &  Clifford subcircuit  & See \cref{algorithm:generalized_CliNR_circuit}\\\hline
          $t_{\ell,j}$& Number of children of $v_{\ell-1,j}$& $t_1$ is the number of\\
 & &vertices at level 1\\\hline
 $\bar t_{\ell}$& $\max_j(t_{\ell,j})$&\\\hline
 $t$& $T(p,n)$&Shorthand in uniformly bounded\\
 & &implementation, \cref{eq:T}\\\hline
 $\alpha$& $18(4\av n +2\bv )+3 \ai n$&$t>1/(\alpha p)$ when $\alpha p<1$\\
 & &See \cref{lemma:logicalbounded}\\\hline
         $s_{\ell,j}$ &  Size of $C_{\ell,j}$ &  ${\bf s}(v_{\ell,j})$ in \cref{algorithm:generalized_CliNR_circuit}, \\
 & &$s=s_{0,0}$\\\hline
         $\bar{s}_{\ell}$&  $\max_j(s_{\ell,j})$& \\\hline
         $ \underline{s}_{\ell} $&  $\min_j(s_{\ell,j})$& \\\hline
         $\ts_{\ell,j}$&  The expected number of gates 
  in the implementation of &\\
 & $\CliNR(\tree_{\ell,j},C_{\ell,j})$&\\\hline
         $\ts_{\ell}$&  $\max_j(\hat{s}_{\ell,j})$& \\\hline
 $r_{\ell,j}$& Number of  stabilizer measurements  at $v_{\ell,j}$ &${\bf r}(v_{\ell, j})$ in \cref{algorithm:generalized_CliNR_circuit}\\
\hline
 $\bar r_{\ell}$& $\max_j(r_{\ell,j})$&\\\hline
 $\underline r_{\ell}$& $\min_j(r_{\ell,j})$&\\\hline
         $\ap$,  $\av$,  &  Constants used to bound gate counts & Independent of \\
  $\bv$,  $\ai$& &$n$, $p$, $s$ and $C$\\
\hline
         $m_{\ell,j}$&  $\ap n+r_{\ell,j}(\av n+\bv)$& no subscript for $\CliNR_1$\\\hline
         $\bar m_{\ell}$&$\max_j(m_{\ell,j})$  & \\ \hline
 $\gp(x)$& $1-(1-p)^x$&\\\hline
 $\plog$& logical error rate &$\plog()$ with argument \\
 & &if context is necessary\\\hline
 $q_{\ell,j}$& Subtree error rate&As in \cref{fig:boundedT}\\\hline
 $\bar{p}_{b\ell}$& Upper bound on $\max_j q_{\ell,j}$.&May not be tight\\\hline
 $\gateoverhead$ & Gate overhead&\\ \hline
 $\qubitoverhead$& Qubit overhead&\\\hline
    \end{tabular}
    \caption{Notation. The following general conventions are used: Maximum over a set is indicated by a bar as in $\bar{r}_{\ell}=\max_j(r_{\ell,j})$; Minimum over a set is indicated by an underline as in $\underline{r}_{\ell} = \min_j (r_{\ell,j}) $
}
    \label{tab:notation}
\end{table*}

\subsection{Bounds on recursive CliNR} \label{sec:technical}

This section provides general bounds on the performance and the overhead of the recursive CliNR scheme, including the statements that lead to \cref{theorem:main}. 

The following known result is used as the motivation for  random stabilizer measurements.

\begin{lemma} \label{lemma:checks} 
Let $\mathcal{S}_r$ be a set of $r$ stabilizers selected independently and uniformly at random from a stabilizer group $\mathcal{S}$. Let $E$ be a Pauli operator that does not commute with $\mathcal{S}$.  The probability that $E$ commutes with $\mathcal{S}_r$ is $2^{-r}$.

\end{lemma}

\begin{proof}
    Exactly half of the stabilizers in the group $\mathcal{S}$ commute with $E$, so each randomly selected stabilizer has a probability of $1/2$ to commute with $E$. The probability that $r$ randomly selected stabilizers all commute with $E$ is therefore $2^{-r}.$
\end{proof}

 Lemma 1 in \cite{delfosse2024low} provides bounds for the performance of $\CliNR_{1,r}$ when the input circuit has a fixed size. The lemma below is a slight extension that explicitly allows the input circuit to have variable size with an expected size $\ts$.  

\begin{lemma}[$\CliNR_1$ bounds] \label{lemma:clinr1}
The implementation $\CliNR_{1,r}(C)$ of an $n$-qubit Clifford circuit $C$ with an expected size $\ts$ has logical error rate 
\begin{align} \label{eq:clinar1_error}
    \plog \le & \frac{\left(1-(1-p)^{\ap n}\left(1-\plog(C)\right)\right)2^{-r}+2\gp(\av n+\bv)}{(1-p)^m(1-\plog(C))}
    \\  \nonumber & +\gp(\ai n).
\end{align}

Moreover, the gate overhead satisfies 
\begin{equation}\label{eq:clinr1_gates}
    \gateoverhead \le \frac{(m+\ts)}{\ts(1-p)^m\left(1-\plog(C)\right)} + \frac{\ai n}{\ts}.
\end{equation} 

    It is possible to implement $\CliNR_{1,r}(C)$ using $3n+1$ qubits. 
\end{lemma}

\begin{proof}
  The fact that $\CliNR_{1,r}(C)$ implements $C$ and that it can be implemented with $3n+1$ qubits is apparent from its structure. Note that the  same ancilla can be re-used on all stabilizer measurements. 

  For $\gateoverhead$, we need to bound the probability of a restart due to a stabilizer measurement.  A necessary condition is an error before or during the stabilizer measurement. The probability of an error  at any point in RSP and RSV is  therefore an upper bound on the restart probability, so $\pres \le 1-(1-p)^m(1-\plog(C))$. 

  The expected number of times the repeating part of the circuit (RSP and RSV) is executed is upper bounded by
  \begin{align*} 
  &\sum_{k\ge 1}k \pres^{k-1}(1-\pres)  \\
  =&\frac{1}{1-\pres}  \\
  \le & \frac{1}{(1-p)^m\left(1-\plog(C)\right)}. 
  \end{align*}

The total expected gate count  is $\gateoverhead \ts \le \frac{(m+\ts)}{(1-p)^m\left(1-\plog(C)\right)} + \ai n$. 

For a logical error at the output of $\CliNR_1$  at least one of the following three events must occur: (1) A logical error during RSP which is not detected during RSV (with probability $p_\rsp$); (2) A logical error during a stabilizer measurement which is not detected by a subsequent stabilizer measurement (with probability $p_\rsv$);  (3) A logical error during RSI (with probability $p_\rsi$). 

The  logical error rate is upper bounded by the sum of the probabilities for the events above, conditioned on the circuit not restarting, 
\begin{equation}\label{eq:sump1p2p3}
    \plog \le \frac{p_\rsp+p_\rsv}{1-\pres}+p_\rsi.
\end{equation}

(RSP) A logical error in RSP will not be detected during RSV if it commutes with all the stabilizer measurements. It might also not be detected due to an error during RSV, but that is accounted for in $p_\rsv$.  
The probability of a logical error in  RSP  is  $1-(1-p)^{\ap n}(1-\plog(C))$.  Using  \cref{lemma:checks} we have $p_\rsp\le \left(1-(1-p)^{\ap n}(1-\plog(C))\right)2^{-r}$

(RSV) The probability of an error during the $i$th stabilizer measurement is $1-(1-p)^{\av n + \bv}$.  The probability that it is a logical error which is also not detected by one of the following stabilizer measurements is at most $\left(1-(1-p)^{\av n + \bv}\right)2^{-r+i}$. So 

\[  
p_\rsv \le \sum_{i=1}^r\left(1-(1-p)^{\av n + \bv}\right) 2^{-r+i} \le
2\gp(\av n + \bv).
\]

(RSI) At most $\ai n$ operations are required for RSI so  $p_\rsi \le \gp (\ai n)$.

Plugging the bounds into  \cref{eq:sump1p2p3} gives \cref{eq:clinar1_error}.
\end{proof}

\begin{lemma}[Logical error bounds for Recursive $\CliNR$] \label{lemma:clinrtlogical}
    Let $\tree$ be a $\CliNR$ tree of depth $D$. The implementation   $\CliNR(\tree,C)$  has a logical error rate   
    \begin{align}\label{eq:clinarT_error}
   & \plog \le  t_1 \Bigg( \\ \nonumber
       & \frac{\left(1-(1-p)^{\ap n}(1-\bar{p}_{b1})\right)2^{-\underline{r}_1}+2\gp(\av n + \bv )}{(1-p)^{\bar{m}_1}(1-\bar{p}_{b1})}\\ \nonumber
         & + \gp(\ai n)\Bigg),
    \end{align}
        
    where 
    \begin{align} \label{eq:pbl}
        &\bar{p}_{b\ell}  =  \bar{t}_{\ell+1}\Bigg(\\ \nonumber
        &\frac{\left(1-(1-p)^{\ap n}(1-\bar{p}_{b{\ell+1}})\right)2^{-\underline{r}_{\ell+1}}+2\gp(\av n +\bv)}{(1-p)^{\bar{m}_{\ell+1}}(1-\bar{p}_{b{\ell+1}})}\\ \nonumber
        &+\gp(\ai n)\Bigg),
    \end{align}
    for $\ell < D$, and

    \begin{equation} \label{eq:pbd}
        \bar{p}_{bD} = \gp({\bar{s}_D}).
    \end{equation} 

\end{lemma}

\begin{proof}
We do the proof by induction starting with $D=1$ as the base case.
At $D=1$ we have a sequence of ${t_1}$ $\CliNR_{1,r_j}(C_{1,j})$. The upper bound of \cref{eq:clinarT_error} is a consequence of adding the logical error rates from  $t_1$ implementations, $\CliNR_{1,r_j}(C_{1,j})$,  and using the bound in \cref{lemma:clinr1}. 

Assume that \cref{eq:clinarT_error} holds for $D=d$ and let us prove this bound for $D=d+1$.
We start by constructing a $\CliNR$ tree $\tree^-$, starting from $\tree$ and removing all vertices at level $D$. By hypothesis,  \cref{eq:clinarT_error} holds for $\CliNR(\tree^-,C)$.  The circuit $\CliNR(\tree,C)$ is constructed (see \cref{algorithm:generalized_CliNR_circuit}) by first  constructing the circuit $\CliNR(\tree^-,C)$ and then replacing each subcircuit $C_{d,j}$ with a sequence of at most $\bar{t}_{d+1}$ $\CliNR_1$ circuits. Each of these has a logical error rate bounded by $\bar{p}_{bd}/\bar{t}_{d+1}$ (from  \cref{lemma:clinr1}), leading to  \cref{eq:clinarT_error}. 

\end{proof}

\begin{lemma}[Gate overhead for $\CliNR$ tree] \label{lemma:clinrtoverhead}
     Let $\tree$ be a $\CliNR$ tree of depth $D$. The gate overhead for the implementation $\CliNR(\tree,C)$ satisfies 
      \begin{equation} \label{eq:clinrT_overhead} 
        \gateoverhead  \le \frac{\bar{t}_1}{s} \left (\frac{(\bar{m}_1+\ts_1)}
        {(1-p)^{\bar{m}_1}(1-\bar{p}_{b1})} + \ai n \right), 
    \end{equation}

     where $\bar{p}_{b1}$ is given by the recursive expressions in \cref{eq:pbl} and \cref{eq:pbd},  and $\ts_1$ is given by the recursive expression 
      
    \begin{equation} \label{eq:sbl}
     \ts_\ell \le \bar{t}_{\ell+1}\left(\frac{\bar{m}_{\ell+1}+{\ts}_{\ell+1}}{(1-p)^{\bar{m}_{\ell+1}}(1-\bar{p}_{b\ell+1})} + \ai n\right)
    \end{equation}
    for $\ell < D$,  and    
    $\ts_D$ is  the size of the largest subcircuit at depth $D$.

\end{lemma}

\begin{proof}
    We prove this by induction.  For $D=1$ (base case) we have a sequence of $\bar{t}_1$ implementations of $\CliNR_{1,r_{1,j}}(C_j)$ each with an upper bound on the expected gate number given by \cref{lemma:clinr1}, leading to  \cref{eq:clinrT_overhead}.

    The induction hypothesis is that \cref{eq:clinrT_overhead} holds for $D=d$.  For $D=d+1$ we begin by constructing a $\CliNR$ tree $\tree^-$ starting from $\tree$ and removing all vertices at level $D$. By hypothesis,  \cref{eq:clinrT_overhead} holds for $\CliNR(\tree^-,C)$. The $\CliNR(\tree,C)$ circuit is constructed by first constructing $\CliNR(\tree^-,C)$ and then replacing each subcircuit $C_{d,j}$ with a sequence of at most $\bar{t}_{d+1}$ $\CliNR_1$ circuits (see \cref{algorithm:generalized_CliNR_circuit}). Each of these $\CliNR_1(C_{D,j})$ circuits has an expected gate count with an upper bound given in  \cref{lemma:clinr1}.  Multiplying the largest of these by $\bar{t}_{d+1}$ gives \cref{eq:clinrT_overhead}. 
\end{proof}

\subsection{Uniformly bounded implementation} \label{sec:bounded}

To prove \cref{theorem:main}, we introduce a new class of recursive $\CliNR$ implementations which we call  \emph{uniformly bounded} (see \cref{fig:boundedT} b). Before going to a formal definition, we provide some intuition by listing a few properties which are valid at the limit $np \rightarrow 0$ (see  proof of \cref{lemma:logicalbounded} and  \cref{theorem:main}). 
\begin{enumerate}

    \item For each implementation of $\CliNR_1$, the logical error rate in $\rsp$ is similar to the logical error rate in $\rsv$, \ie the two terms in the numerator of \cref{eq:pbl}  have a similar value (of order $np$). 

\item There is a bounded subtree error $q_{\ell,j}<2/3$ for all $\ell \ge 1$.

\item The logical error rate $ q_{0,0} < O(sn^Dp^{D+1})$. 

    \item The logical error rate $q_{0,0}$ vanishes when $D=\max\left(1, \log(sp)\right)$.
\end{enumerate}

Property 1 means that the logical error at the output of each $\CliNR_1$ block is independent of $s$. Intuition on  property 2  and how it leads to 3 is given in \cref{fig:boundedT} (b). Property 4 follows from 3 (\cref{theorem:main}). 

\begin{definition} [Uniformly bounded $\CliNR$] 
\label{def:bounded} 
A uniformly bounded $\CliNR$ implementation with depth $D$ is a recursive $\CliNR$ implementation using a $\CliNR$ tree constructed using \cref{algo:bounded}.

For the purposes of the algorithm we define the following functions
 \begin{equation}\label{eq:T}
    T(p,n) := \left \lfloor\frac{2}{18(4\av n +2\bv)p+3 \ai n p} \right \rfloor
    \end{equation}
and 
\begin{equation}\label{eq:R}
R(p,n) :=
    \lceil \log\left(p\ap n\left(1-\frac{2}{3}\right)+\frac{2}{3}\right)-\log(2 \av n p)\rceil.
 \end{equation}

We say that $\tree({\bf s}, {\bf r})$ is empty if it has no vertices and the functions ${\bf s}$ and ${\bf r}$ map all vertices $v_{\ell,j}$ to 0.

\begin{algorithm}
\caption{Algorithm for constructing a $\CliNR$ tree for a uniformly bounded implementation}
\label{algo:bounded}
\DontPrintSemicolon
\SetKwInOut{Input}{input}
\SetKwInOut{Output}{output} 
    \Input{An $n$-qubit Clifford circuit $C$ with size $s$. 
    A noise parameter $p$.}
    \Output{A $\CliNR$ tree $\tree$.}

    Let $R=R(p,n)$, and $T=T(p,n)$.

Create an empty tree $\tree = \tree({\bf s}, {\bf r})$. All vertices added below are for $\tree$. 

Set $t'=\lceil \frac{3sp}{2}\rceil$.

Add vertices $v_{D,j}$, $j\in \{0,\dots, t'-1\}$.

\For {$j \in \{0,\dots ,(s \mod t')-1\}$}{ 
    set  ${\bf s}(v_{D,j}) = \lceil \frac{s}{t'} \rceil$.
    
    }
\For {$j \in \{s \mod t' ,\dots, t'-1\}$}{ 
    set  ${\bf s}(v_{D,j}) = \lfloor \frac{s}{t'} \rfloor$.
    
    }

\For {  $ \ell \in \{D-1,\dots, 1\}$ }{
Set $k=0$.

\While{there are vertices in layer $\ell+1$ with no parent in layer $\ell$}{

Let $V$ be the set of first $T$ vertices in layer $\ell+1$ with no parent.

Add vertex $v_{\ell,k}$.

Make all vertices in $V$ children of $v_{\ell,k}$.

Set ${\bf s}(v_{\ell,k})=\sum_{v\in V}{\bf s}(v)$.

Set $k=k+1$.
}}

For all vertices $v_{\ell,j} \in \tree$, set ${\bf r }(v_{\ell,j})=R$. 

Add vertex $v_{0,0}$ and set ${\bf s}(v_{0,0})=s$, and ${\bf r}(v_{0,0})=0$.

Make all vertices at level one children of $v_{0,0}$.
    
    \Return{$\tree$. }
\end{algorithm}

\end{definition}

The following properties are a direct consequence of \cref{algo:bounded}.

\begin{proposition}
    
The uniformly bounded implementation of a Clifford circuit with $s>2/p$ has the following properties:
\begin{itemize}
    \item $C_{D,j}$ are of size at most $\bar{s}_D = \frac{2}{3p} $ and at least $ \frac{1}{2p}$.

    \item The $\CliNR$ tree has $\lceil\frac{s}{\bar{s}_D}\rceil$ vertices at level $D$.

    \item Each vertex $v_{\ell,j}$ with $D>\ell>0$ has  at most $t=T(p,n)$  children. 
    
      \item  $r_{\ell,j}=R(p,n)$ for all $\ell\ne 0 ,j$.

   \end{itemize}

\end{proposition}

\begin{lemma} [Logical error rate  for a uniformly bounded  implementation] \label{lemma:logicalbounded}
Using $\alpha = 18(4\av n +2\bv )+3 \ai n$,
if $\ap n+ 2R(p,n)(\av n +\bv) +\ai n  < 1/(2p)$ and $\alpha p <1$, then  a uniformly bounded $\CliNR$ implementation with depth $D$ has a logical error rate

\begin{equation} \label{eq:logicalerrorboundedtree}
\plog \le
\begin{cases}
2sp^{D+1}\alpha^{D},
 & \left(\frac{1}{\alpha p}\right)^{D-1} < \frac{3sp}{2},\\[3ex]
\dfrac{2\alpha p}{3},
 & \left(\frac{1}{\alpha p}\right)^{D-1} \ge \frac{3sp}{2}.
\end{cases}
\end{equation}
\end{lemma}

 \begin{proof}
For a uniformly bounded implementation, the number of children for  $v_{0,0}$ is $t_1 \le  \left\lceil \frac{s}{\bar{s}_D t^{D-1}}\right \rceil$.  The value of all $m_{\ell,j}$ is fixed and denoted $m$.

Using \cref{lemma:clinrtlogical}, our aim is to prove that $\bar{p}_{b\ell}$ in  \cref{eq:pbl} remains bounded for all $\ell>1$. We do this by first assuming that $\bar{p}_{b\ell+1} \le 2/3$ and then proving that this assumption is true for all $\ell$.

Using  the Bernoulli bound $(1-p)^x \ge (1-xp)$ in \cref{eq:pbl} gives,  
 \begin{align} 
        \bar{p}_{b\ell} \le  & \bar{t}_{\ell+1}\Bigg(\\ \nonumber
        &\frac{(\ap n p(1-\bar{p}_{b{\ell+1}})+ \bar{p}_{b{\ell+1}})2^{-{r}}+2(\av n +\bv)p}{(1-mp)(1-\bar{p}_{b{\ell+1}})}\\ \nonumber
        &+\ai n p\Bigg). 
    \end{align}
Since $ r=\lceil \log(p\ap n(1-2/3)+2/3)-\log(2 \av n p)\rceil$  and assuming  $\bar{p}_{b{\ell+1}} \le 2/3$ gives  
\[ \bar{p}_{b\ell} \le \bar{t}_{\ell+1}  \left( \frac{(4\av n +2 \bv)p}{{(1-mp)(1-\bar{p}_{b{\ell+1}})}} +\ai n p \right).  \]

By assumption $1-mp >1/2 $, and by construction $\bar{t}_{\ell+1}< \frac{2}{18(4\av n + 2\bv)p+3 \ai n p}$ so if $\bar{p}_{b{\ell+1}} \le 2/3$ we have  $\bar{p}_{b\ell} \le  2/3 $. Since $\bar{p}_{bD}<2/3$ is ensured from  $\bar{s}_D p \le 2/3 $ then our assumption $\bar{p}_{b\ell} \le 2/3$  is guaranteed to be true for all $\ell \ge 1$. 

 Doing the same for \cref{eq:clinarT_error},  replacing $t$ with $t_1$ gives  $\plog \le \frac{2}{3}\frac{t_1}{t}$.  Using $t > 1/(\alpha p)$ we derive \cref{eq:logicalerrorboundedtree} by taking two cases: When $t_1=1$, which occurs when $t^{D-1}\ge s/\bar{s}_D = 3sp/2$, we get $\plog \le \frac{2}{3t} \le \frac{2 \alpha p}{3}$. Otherwise, $t_1>1$ and $t_1 = \left\lceil \frac{s}{\bar{s}_D t^{D-1}}\right\rceil < \frac{2s}{\bar{s}_D t^{D-1}}$, so substituting $\bar{s}_D = \frac{2}{3p}$ gives $\plog \le \frac{4s}{3\bar{s}_D t^{D}} = \frac{2sp}{t^{D}} \le 2sp^{D+1}\alpha^{D}$.
\end{proof}

Note that at $D=1$ we recover the $snp^2$ scaling  obtained in \cite{delfosse2024low}.
When $D > 1$, we get $\plog \leq O(sn^D p^{D+1})$ up to the point $t_1=1$ where the bound saturates at $\plog \le 2 \alpha p /3$.

\begin{lemma} [Gate overhead for a uniformly bounded  implementation] \label{lemma:gateoverheadbound}

Under the same conditions as  \cref{lemma:logicalbounded}, $\gateoverhead$  for $\CliNR(\tree^{B,D}, C)$ satisfies 
\[ \gateoverhead  \le 2(12^{D})  .
\]
\end{lemma}

\begin{proof}
    Starting with the bounds on $\gateoverhead$ in \cref{lemma:clinrtoverhead} and plugging $\bar{p}_{b\ell}<2/3$  in  \cref{eq:sbl}, together with  the condition  $mp< 1/2 $ and the Bernoulli bound we get,

   \[ \ts_\ell \le t_{\ell+1} \left ( 6(m+\ts_{\ell+1}) + \ai n \right).  \]
    We also have $\bar{s}_D \le  \frac{2}{3p}$ by construction. By assumption  $m +\ai n < \frac{1}{2p}$ reducing the expression to 
    
    \[ \ts_\ell \le t_{\ell+1}\left ( 6(\frac{1}{2p}+\ts_{\ell+1})\right ). \] 

    This implies $\ts_\ell \le 12 t_{\ell+1} \ts_{\ell+1}$ since  $\ts_{\ell+1} \ge  \frac{1}{2p} $. We then get  
    $\gateoverhead s \le \prod_\ell (12t_\ell)  \bar{s}_D \le  12^{D} 2 s$, where in the last step we used the fact that in the uniformly bounded implementation  $\prod_\ell t_\ell = \left\lceil \frac{s}{\bar{s}_D}\right \rceil\le2s/\bar{s}_D$.    
\end{proof}

\subsection{Tightness of the bounds}\label{sec:tightness}
The bounds given in \cref{theorem:main} and derived above are  meant to show that $\gateoverhead$ in recursive $\CliNR$ is polynomial in $sp$ and that $\qubitoverhead$ is logarithmic. 
They are, however, not tight. This can easily be seen by comparing with the results of \cite{delfosse2024low}. Below we note some of the points where loose bounds were used. These compound in the final result. 

\begin{itemize}
    \item The restart probability is upper bounded by the probability that the algorithm restarts on every fault. 
    \item The depth of the tree $D$ is upper bounded by the worst case depth when $np\rightarrow 0$. From \cite{delfosse2024low} we know that $D=1$ suffices when $snp^2 \rightarrow 0$.
    \item We sum over probabilities for independent events to bound the probability that at least one event occurs. 
    \item Repeated use of the Bernoulli bound.
    
\end{itemize}

It should also be noted that the uniformly bounded implementation is derived to be simple rather than optimal.  

To get a better intuition of where recursive $\CliNR$ could be applied in a realistic scenario, we turn to numerical results. These results include a search for optimal tree parameters and allow for a range of different implementations. We do not attempt to recover the results of the uniformly bounded implementation since it is unlikely to be optimal.

\section{Numerical results} \label{sec:numerical}
To show the advantage of Recursive $\CliNR$ in experimentally relevant regimes we use two different techniques: A fast Markov model technique 
(see \cref{sec:estimates})  and a Monte Carlo simulation.  The Markov model (based on \cite{ibmflags}) works by tracking the probabilities of errors and stabilizer measurement outcomes through the circuit. It is significantly faster than the  simulation (built using Stim \cite{Gidney_Stim_a_fast_2021}), but also less accurate (see \cref{fig:analyt_vs_sim_plog_1} in \cref{app:estimates}).  In both cases we consider the circuit level noise model with error rate $p$ for two-qubit gates, $p/10$ for single-qubit operations (gate and measurements) and either $p/1000$ for idle qubits or no idle errors. The errors are  depolarizing for qubits and bit-flip for measurements.

For the simulation, we use $p=10^{-3}$. For the Markov model we use both $p=10^{-3}$ and $10^{-4}$ and supplement the error model with estimated parameters given  in \cref{sec:example}. Results are shown in \cref{fig:estimated-performance} (Markov model) and \cref{fig:pareto} (Simulation). 

The values were chosen based on existing theoretical and experimental results for trapped ions.  One- and two-qubit errors  below $0.002\%$ and $0.01\%$ respectively have been demonstrated  with trapped ions \cite{oxfordgates,hughes2025trapped}. Ion qubit coherence times of $T_2^*>1000$ seconds have been demonstrated \cite{Wang_2021}, and theoretical results predict fast two-qubit gates for trapped ions  on the order of a microsecond \cite{fastgates}. At these theoretically achievable gate speeds and demonstrated $T_2^*$, the idle error rate should be lower than $10^{-9}$, three orders of magnitude lower than our value for the $p=10^{-3}$ simulations.  

The advantages of Recursive $\CliNR$ are expected to be more pronounced for large circuits. We therefore studied the family of circuits that has the largest expected gate-count (for given $n$), \ie random Clifford circuits, making slight adjustments (either trimming or padding) so that the gate count is fixed at $n^2$. 

\subsection{Estimated performance}
Using the Markov model (\cref{sec:estimates}) we searched for values of $n$  where Recursive $\CliNR$ at depth 2 will outperform the standard $\CliNR$ scheme (\ie Recursive $\CliNR$ at depth 1).  Results presented here are for 70 qubits at $p=10^{-3}$ with idle noise (\cref{fig:estimated-performance} (a)) and 400 qubits at $p=10^{-4}$ with no idle noise (\cref{fig:estimated-performance} (b)). We considered trees with different numbers of vertices at $\ell=1$ (ranging from 1 to 10), each having the same number of children (0 for $D=1$, and 2 to 10 for $D=2$).  Each tree had a fixed $r_{\ell,j}=r_\ell$ for $\ell \ne 0$. We used $r_\ell$ ranging from 0 to 30.  

The plots (\cref{fig:estimated-performance}) show the Pareto frontier (best $\gateoverhead$, $\plog$ values) at each depth with a cutoff for the overhead, $ \gateoverhead \le 100$. We note that if we allow arbitrarily large $\gateoverhead$, it is possible to reach much lower logical error rates at both depth 1 and depth 2.

The results of the estimates should be used  as a guide for showing relative performance qualitatively, and for finding good tree parameters. 
The precision drops as the overhead increases but improves when there is no idle noise.  
These limitations are discussed in \cref{app:estimates}. 

At $n=70$ (\cref{fig:estimated-performance} (a)) we see that Recursive $\CliNR$ has a significant advantage at around $\gateoverhead = 21$, providing a guide for points we want to simulate (\cref{fig:pareto}). At $n=400$, \cref{fig:estimated-performance} (b)  shows a significant advantage for the depth 2 tree starting at around $\gateoverhead =10 $. We also see that logical error rates  plateau at various points both at depth 1 and 2 and then start improving again. Escaping the plateau requires a tree with fewer vertices. The cost is a (usually significant) increase in  gate overhead.  We note that this continues at larger overheads (not shown) and that both depth 1 and 2 seem to have the same performance at very large $\gateoverhead$.

\begin{figure}
    \centering
        \includegraphics[width=0.7\linewidth]{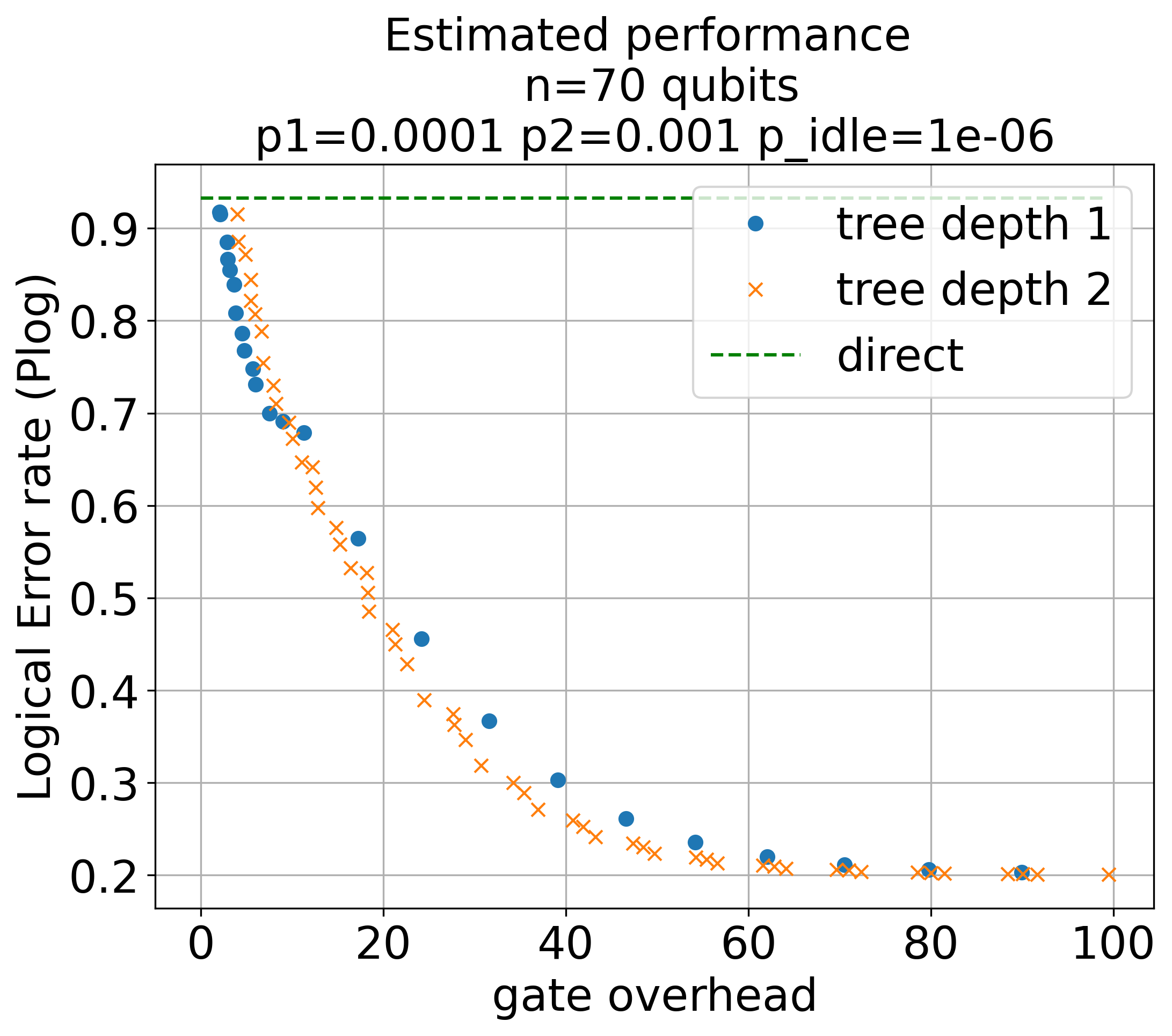} 

       \hspace{10pt} (a)

       \vspace{10pt}
        
    \includegraphics[width=0.7\linewidth]{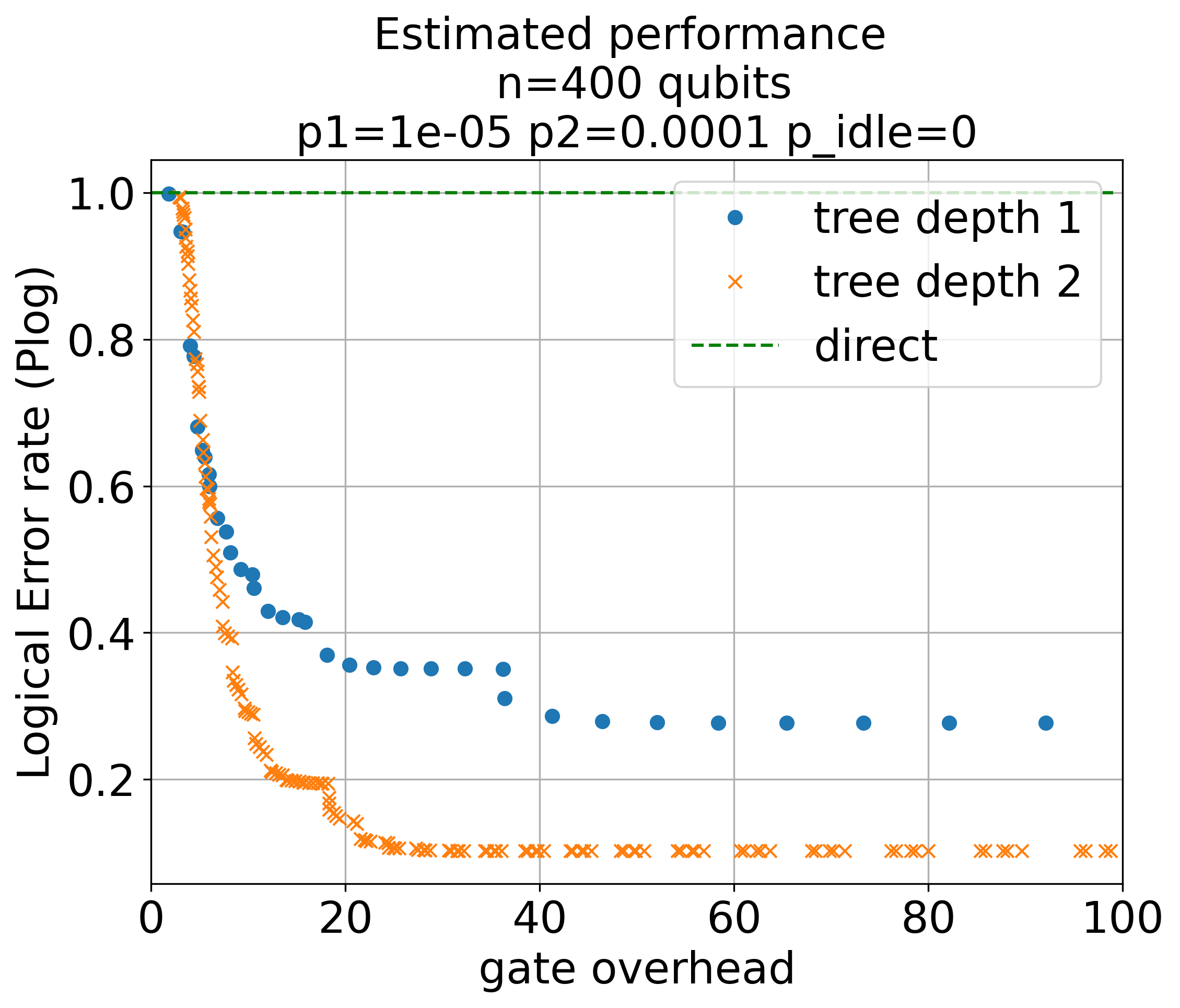}

  \hspace{10pt}  (b)

    \caption{Numerical estimates of the performance based on the Markov model at (a) $n=70$, $p=10^{-3}$ with idle noise and (b)  $n=400$, $p=10^{-4}$ with no idle noise. In both cases the gate overhead was capped at $\gateoverhead < 100$. Plotted points are the Pareto frontiers at depth 1 and 2.  Note that the performance estimates tend to be less accurate when  the gate overheads are large but the trend remains indicative (see \cref{app:estimates}). Results for $n=70$ can be compared with direct simulation in \cref{fig:pareto}. Results for $n=400$ show a significant advantage when going to depth 2, {\it e.g.} at a gate overhead of $\gateoverhead \approx 25.5$ we have $\plog \approx 0.35$ at depth 1 and $\plog \approx 0.10$ at depth 2.  }
    \label{fig:estimated-performance}
\end{figure}

\subsection{Numerical simulations} \label{sec:simulations}
We used Stim \cite{Gidney_Stim_a_fast_2021} to run  Monte Carlo simulations.  These  included a complete run of the circuit with all restart events. Random Cliffords were generated using Qiskit \cite{Bravyi_2021} with some padding or truncation of the generated circuit so that the gate count is fixed at $n^2$. Each data point in \cref{fig:pareto} is the average over $50$ different Clifford input circuits, each with $80$ shots. The input state was always  $\ket{0}^{\otimes n}$. Error bars are the standard deviation over the $50$ different circuits. 

Using the Markov model, we were unable to identify an advantage for Recursive $\CliNR$ at $n<100$ and  $p=10^{-4}$. Consequently we set $p=10^{-3}$ and $n=70$ to  keep the simulation time manageable.  Results are plotted in  \cref{fig:pareto} and compared to the estimates of the previous section in \cref{app:estimates}. Specific data points were chosen by using the Markov model to generate a Pareto frontier (see \cref{tab:treeparams,tab:treeparams-noidle}) and taking points where $\gateoverhead \le 21$. 

The results show a similar trend to the estimates in \cref{fig:estimated-performance}. When there is no idle noise, we clearly see the advantage of  using depth 2 trees (\ie Recursive $\CliNR$) starting at  $\gateoverhead \approx 15 $. With idle noise, the advantage is less obvious due to inaccuracies introduced by choosing points with the Markov model (see \cref{fig:analyt_vs_sim_plog_1}).  Nevertheless, we  can see  the improvement in depth 2 at the larger values of $\gateoverhead$ as predicted in \cref{fig:estimated-performance} (a).   We also see the impact of idle noise when comparing \cref{fig:pareto} (a) to (b). 

A comparison of the Stim simulation and the Markov method estimates is given in 
\cref{app:estimates}.  

\begin{figure}
    \centering
    \includegraphics[width=0.9\linewidth]{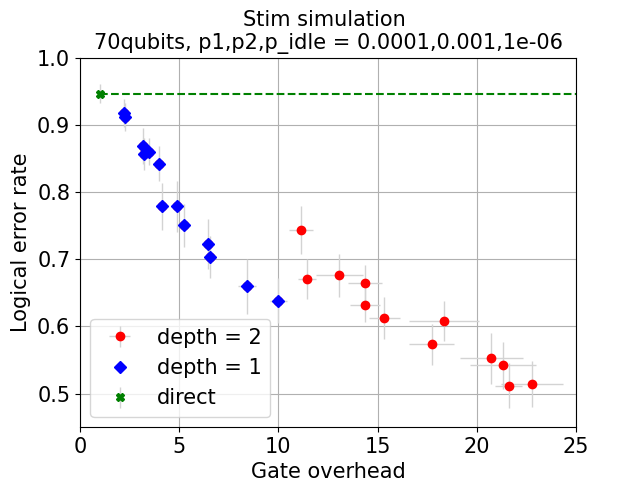}
    
       \hspace{20pt} (a)

       \vspace{10pt}
     
    \includegraphics[width=0.9\linewidth]{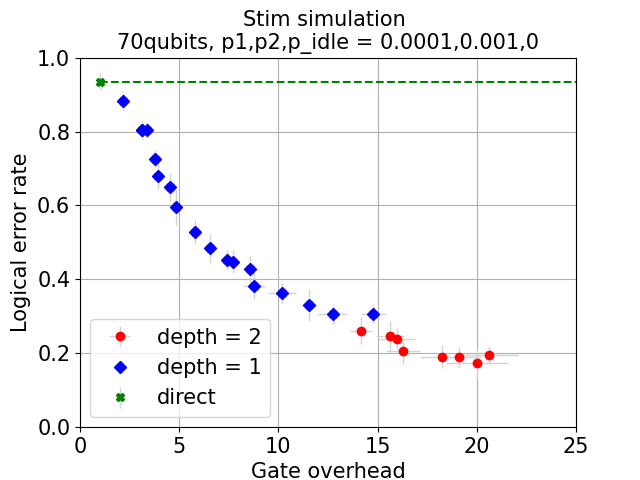}

       \hspace{20pt} (b) 
        
    \caption{Simulation results at $n=70$ and $p=10^{-3}$. (a) With idle noise; (b) Without idle noise.  Points were selected using the Markov model with $\gateoverhead <21$ (see \cref{tab:treeparams,tab:treeparams-noidle}).
    In both cases we see the trend predicted in \cref{fig:estimated-performance} as we move from depth 1 (blue diamonds) to depth 2 (red circles). The direct implementation is given as a reference (green dashed line).}
    \label{fig:pareto}
\end{figure}

\section{Conclusions}
\label{sec:conclusions}
Our main result shows that Recursive $\CliNR$ can be used to reduce the logical error rate of any size circuit to an arbitrarily low value when $np\rightarrow 0$. This result can be compared with the asymptotic bound   $snp^2 \rightarrow 0$ for the original $\CliNR$. For the largest Clifford circuits (of order $n^2/\log(n)$ \cite{Aaronson_2004}) the condition for  $\CliNR_1$ becomes  $(n^3p^2)/\log(n) \rightarrow 0$ compared to $np \rightarrow 0$ for Recursive $\CliNR$.  Putting numbers in this simplified model, at $p=10^{-3}$ we expect recursive $\CliNR$ to be relevant at around $n=100$.

Our numerical results show that in practice (assuming realistic physical error rates and near-term qubit counts)  the main advantage of the recursive scheme is at fixed gate overheads.  Recursive $\CliNR$ can have lower logical error rates when we limit $\gateoverhead$.  The advantage is particularly pronounced when the idle errors are very small, a feat that is achievable for trapped ion qubits.

 If more qubits are available, it is possible to parallelize operations, in particular the sequential implementation of $\CliNR_1$ blocks at the same depth. This is discussed in a follow up to this work \cite{dobbs2025advantage}.  Additionally it is possible to replace random stabilizers with stabilizers specifically chosen for the desired circuit, noise model and device restrictions (as in \cite{tham2025optimized}) and reduce the overheads in RSV which would subsequently reduce idle time.  

 In the near-term we expect $\CliNR$, Recursive $\CliNR$ and similar schemes to play an important role in quantum computing. Both techniques can be applied mid-circuit, following and preceding non-Clifford blocks. There is already evidence that Clifford circuits conjugated by single-qubit non-Clifford gates can be used for quantum supremacy experiments \cite{bouland2017complexity}.  It is reasonable to expect that some of the so-called utility-scale algorithms could be compiled to include large Clifford blocks that may be too deep to apply directly on physical qubits, but could be implemented with $\CliNR$ or recursive $\CliNR$ with reasonable overheads (see for example \cite{brown2026mid}). In the longer term it is likely that these types of schemes will still play a role in conjunction with QEC, maximizing resource efficiency (see  \cite{webster2026fast} for a use-case with qLDPC codes).

\acknowledgments{
The authors thank Felix Tripier, Andy Maksymov, Min Ye, John Gamble and the whole IonQ team for insightful discussions. }

\section*{Author contributions}
All authors participated in the research. AB and GB wrote the code. AB and ND wrote the manuscript. 
As part of our revisions following referee comments we used AI to review and critique mathematical statements, review the code and ensure that the description of the Markov simulation is aligned with the code. 

\bibliographystyle{quantum}

\bibliography{references}

\appendix

\section*{Appendix}

\crefalias{section}{appendix}
\crefalias{subsection}{appendix}
\section[appendix]{Optimizing the tree structure} \label{sec:optimizing}

Our aim is to provide practical methods for optimizing the tree parameters. In particular, we are interested in finding a good solution to the following task: {\it Given an input circuit $C$ and a noise model, find a tree that optimizes the values of $\plog$ and $\gateoverhead$ in  the Recursive $\CliNR$ implementation. } Here \emph{optimal} means that the results are on the Pareto frontier. In practice, we would usually also want to bound the gate overhead ({\it e.g.}  $\gateoverhead\le 100$  in \cref{fig:estimated-performance}). The methods below are heuristic and are expected to provide close to optimal solutions. 

Naively, one can simulate performance for different tree parameters and search for those that provide the best performance. In practice, the number of possible tree parameters is large, and  Monte Carlo simulations such as those used in \cref{sec:simulations} are relatively slow. The Markov model below is less accurate than a simulation (see \cref{app:estimates}) but is significantly faster. It can be used to narrow down the search space.

\subsection[appendix]{Markov model and transition matrices} \label{sec:estimates}
We use a modification of the Markov model introduced in \cite{ibmflags} and a simplified error and circuit model to estimate errors at each stage. 

The Markov model uses transition matrices to update a probability vector. We use the notation $T_{j,k}$ to denote the entry in the $j$th row and $k$th column of a transition matrix $T$.  We start with estimating performance for $\CliNR_{1,r}$ and then show how the model can be modified for recursive $\CliNR$. 

The Markov model works by propagating a vector $\vp$ of length $r+2$ that represents the probabilities for various events: $\vp_0$, no error; $\vp_1$, undetected error; and $\vp_{k+2}$, detection at stabilizer measurement $k \in \{0, \dots,r-1\} $.  

{\it The  RSP probability vector:} We denote $p_p$ as the logical error rate for RSP.  The state $\vp$ after RSP has  $\vp_0 = 1-p_p$, $\vp_1=p_p$ and $0$ for all other entries. 

{\it The RSV transition matrices:} We use the stochastic matrix $T(k)$ defined below to propagate $\vp$ through the $k^\text{th}$ stabilizer measurement. 
$T_{00}(k) = 1-p_{de}-p_{ue}$ where $p_{de}$ and $p_{ue}$  are the probabilities for errors during the stabilizer measurement (including measurement faults) that will be detected or undetected respectively.  Similarly $T_{k+2,0}(k) = p_{de}$ and $T_{10}=p_{ue}$.

If there are no errors during the stabilizer measurement, the probability that a previously undetected error will be detected is $1/2$, so $T_{1,1}(k)=1/2$, $T_{k+2,1}(k) = 1/2$. We also have $T_{j,j}(k)=1$ for all $j >1$. 

{\it The RSI transition matrix:} we can use a similar matrix $W$ to propagate through RSI. Setting $p_I$ as the logical error probability of RSI we have $W_{00}=1-p_I$, $W_{10}=p_I$ and $W_{k,k}=1$ for all $k>0$. 

{\it Calculating $\plog$ and $\gateoverhead$:} 
The logical error rate  $\plog'$ following this process is calculated using $\plog' = \frac{\vp_{1}}{\vp_0+\vp_1}$. 

{\it Gate count:} If we have  $g_P$ gates for RSP,  $g_C$ gates for each stabilizer measurement with $r'$ stabilizer measurements and $g_I$ gates for RSI then the expected total number of gates executed is given by 
\begin{equation} \label{eq:gtot}
    g_{tot} = g_P+r'g_C+g_I + \sum_{k=0}^{r'-1}(g_P + (k+1) g_C) m_{\restart(k)} 
\end{equation} 
where $m_{\restart(k)}$ is the expected number of restarts at the $k^\text{th}$ stabilizer measurement, $m_{\restart(k)}=\frac{\vp_{k+2}}{\pres}\left(\frac{1}{1-\pres}-1\right)$ and $\pres = \sum_k\vp_{k+2}$.

{\it Recursive $\CliNR$:} For  Recursive $\CliNR$  we propagate the vector through the circuit, re-initializing $\vp$ after each RSI step. We use  the logical error rate calculated in the previous implementation, $\plog'$,  to update the injection error $p_I$ in a way that depends on the error model (see below) and specifically also accounts for idle errors. 

The last step makes it easy to account for idle errors on the input rail (which can have significant impact). When we go one level up we set $\vp_{0}=1-\plog'$ and $\vp_{1}=\plog'$ and follow the same procedure (see \cref{fig:markov}).

\begin{figure}
    \centering
    \includegraphics[width=1\linewidth]{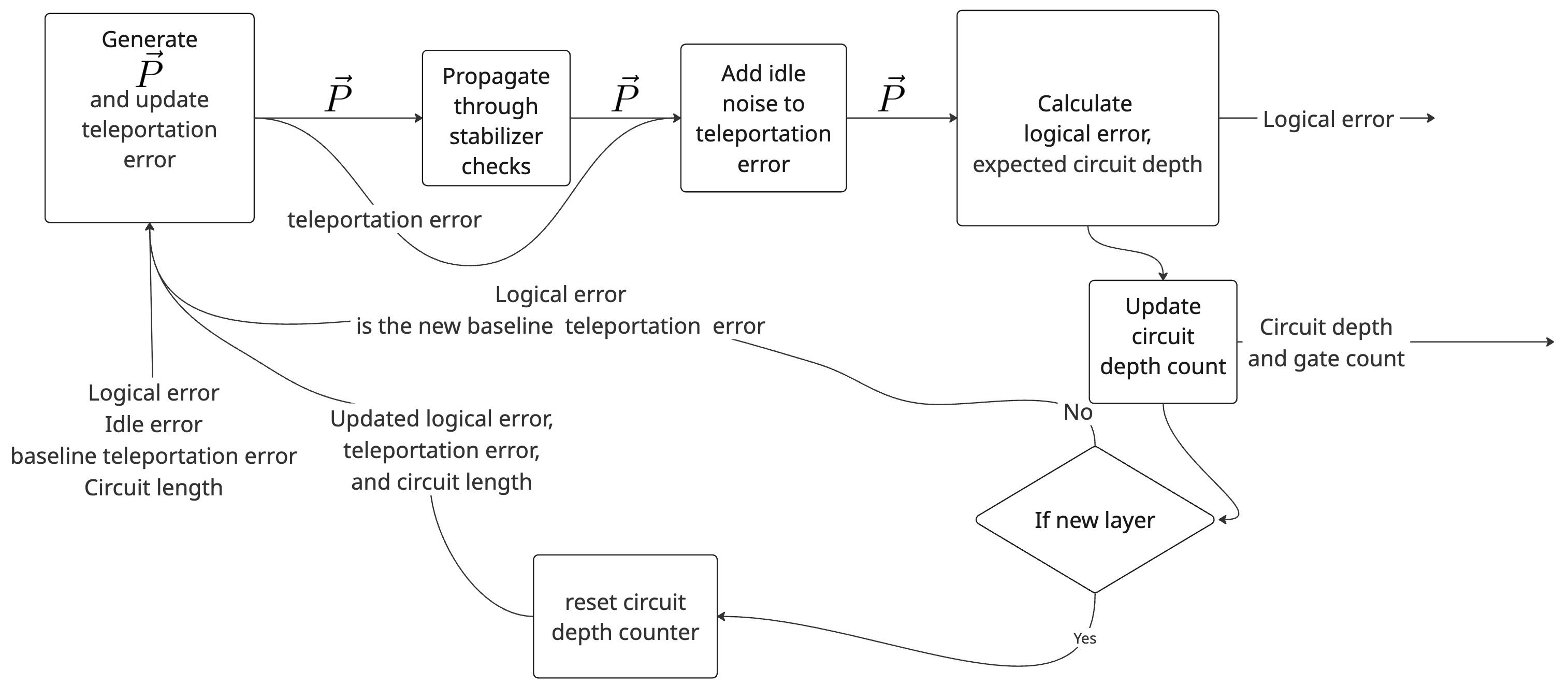}
    \caption{Estimating the logical error rate and gate overhead. The probability vector $\vp$ is initialized with parameters from $\rsp$ and propagated through the stabilizer measurements ($\rsv$). Idle noise on the input rails is accounted for during $\rsi$ and the logical error rate and expected gate count are calculated. The process is then repeated for the next $\CliNR$ block.  At each new level we reset the depth-counter and update the expected depth of the input circuit.}
    \label{fig:markov}
\end{figure}

\subsection[appendix]{Implementing different error models}
The error model depends on two main factors. The implementation of each stage ({\it e.g.} the gate-set, connectivity, parallel operations) and the noise sources. Our aim is to provide values for $p_p$, $p_{de}$, $p_{ue}$ and $p_I$, noting that these (in particular $p_p$ and $p_I$) change as we go through the different instances of $\CliNR_1$. This is because $p_I$ depends on the input state which becomes more noisy as we go through implementations at the same level, and $p_p$ depends on the input circuit which becomes more noisy as we go up levels. 

The preparation and verification errors $p_p$, $p_{de}$, $p_{ue}$ need to account only for the qubits in the resource state.  The RSI error, $p_I$, includes all errors in the input rails, including  $\plog'$ from the previous circuit, and any idling errors during execution. The idling errors are particularly sensitive to the expected size of the circuit which depends on the detection probabilities $\vp_{k+2}$ (see \cref{app:estimates}). 

\subsection[appendix]{Example} \label{sec:example}

The following example provides details of how we used the Markov model for the results presented here (\cref{fig:estimated-performance,fig:analyt_vs_sim_plog_1,tab:treeparams,tab:treeparams-noidle}). 

We consider the circuit level noise model with error rate $p$ for two-qubit gates, $p/10$ for single-qubit operations (gate and measurements) and either $p/1000$ for idle qubits or no idle noise. 
The errors are  depolarizing for qubits and bit-flip for measurements.

We also use the following approximations which are derived from averages over numerical data. The level-$D$ subcircuits $C_{D,j}$ of size $s'$ have $s'/2$ two-qubit gates and $s'/2$ single-qubit gates; the total number of single-qubit and two-qubit gates for RSP are $s'/2+2n$ single-qubit gates and $s'/2+n$ two-qubit gates;  the main source of idle errors is from qubits not participating in $C_{D,j}$ and  the number of idle gates (in RSP) is  $\mathcal{I}_D = s'n/3$.  For completeness we also include the idles during Bell state preparation. 

The error rate for RSP is 

$$p_p=1-(1-p)^{s'/2+n}(1-p/10)^{s'/2+2n}(1-p/1000)^{s'n/3+n}.$$

At levels above $D$ we have $$p_p = 1-(1-\plog')(1-p)^n(1-p/10)^{2n}(1-p/1000)^{n}.$$ Note that the idle errors due to the implementation of the previous layer will be accounted for in the teleportation step. 

For RSV we assume that each stabilizer measurement requires $3n/2$ two-qubit gates (a random controlled $2n$-qubit Pauli implemented using a sequence of controlled $X$, $Y$ or $Z$), two single-qubit gates and a measurement.   There are 15 different two-qubit errors but one has no impact ($X$ error on an $X$ eigenstate), eight impact the measurement of the ancilla (detectable errors) and six do not. There are also $3n^2$  idle gates whose errors are detectable $1/4$ of the time (the error needs to anticommute with the stabilizer and happen before the relevant gate). Putting these together gives $$p_{de}=1-(1-8p/15)^{3n/2}(1-p/10)^{3}(1-p/4000)^{3n^2},$$ and $$p_{ue} = 1-(1-6p/15)^{3n/2}(1-3p/4000)^{3n^2}.$$

For RSI we have a total of $n$ two-qubit gates and $4n$ single-qubit gates or  measurements. We assume that we minimize idling errors by timing the first operation in every branch correctly. When  multiple $\CliNR_1$ circuits feed into the same parent (\ie the corresponding $v_{\ell,j}$ have the same parent $v_{\ell-1,j'}$) then the first of these operations will have
$$p_I=1-(1-p)^{n}(1-p/10)^{4n}(1-p/1000)^{4n}.$$
Subsequent rounds must also account for idle errors on the $n$ input rails. We can use the same reasoning used to derive  \cref{eq:gtot}. At level $D$  the number of idle gates on the input rails from RSP remains $\mathcal{I}_D = s'n/3$. At each stabilizer measurement we have  $3n^2/2$ idle gates on the input rails. That means that the expected total number of idle gates for one block is
\begin{align}
\mathcal{I}_{\ell} &= t_{\ell+1}\Big(\mathcal{I}_{\ell+1} + r'3n^2/2 \\ \nonumber
&+\sum_{k=0}^{r'-1} \left (\mathcal{I}_{\ell+1} + (k+1)(3n^2/2)\right)m_{\restart(k)}\Big).
\end{align}

Finally, $$p_I = 1-\left(1-\plog'\right)(1-p)^{n}(1-p/10)^{4n}(1-p/1000)^{\mathcal{I}_\ell+4n},$$
where $\plog'$ is the logical error rate of the preceding implementation.

\subsection[appendix]{Optimizing the tree}
As noted previously, our aim is to find the optimal tree. This can be done by first iterating over a very large number of tree parameters using the Markov method and finding the Pareto frontier (as in \cref{fig:estimated-performance}). Once we identify a good parameter range, we can use a slower, but more accurate method to find the best tree parameters around that range.

\section[appendix]{Analysis of the numerical estimates} \label{app:estimates}

The Markov method introduced in \cref{sec:estimates} is expected to provide an estimate of the logical error rates and gate overheads which can be used in parameter selection. The use of various approximations in the error parameters is expected to lead to small deviations. The small deviations can compound as the tree becomes deeper, leading to less accurate results at larger $D$. Below we compare the results of the Markov model to the results from the stabilizer simulation (see \cref{sec:simulations}). 

 \cref{fig:analyt_vs_sim_plog_1} shows the comparison between the two methods with and without idle noise. In both cases the parameters are $p=10^{-3}$ and $n=70$. All other parameters are chosen as in \cref{sec:numerical}.

Results with no idle noise \cref{fig:analyt_vs_sim_plog_1} (a), (b) have good agreement between the two methods. When idle noise is added, \cref{fig:analyt_vs_sim_plog_1} (c), (d), the logical error rates are still in good agreement, but the overhead estimates using the Markov method are  consistently lower than the simulation when the values are high. The error bars in simulation are also much larger when $\gateoverhead$ is large. 

The two methods have good qualitative agreement, and in most cases agree on the relative performance of any two data points. There are however some discrepancies, particularly in the plots with idle noise. The result is that some data points are not in the Pareto frontier of the Stim simulation. This is particularly apparent as we switch from depth-1 to depth-2.

One expected imperfection in the application of the Markov method is the use of bulk error and gate-count estimates. The method can be improved by using the exact Clifford circuits and sub-circuits to find gate counts and idle times. Other issues such as accurate error propagation are difficult to avoid without a direct simulation. The qualitative result should be sufficient to provide guidance for narrowing down the tree parameters before going into more accurate methods. 

\begin{figure*}[h]
    \centering

    \includegraphics[width=0.48\linewidth]{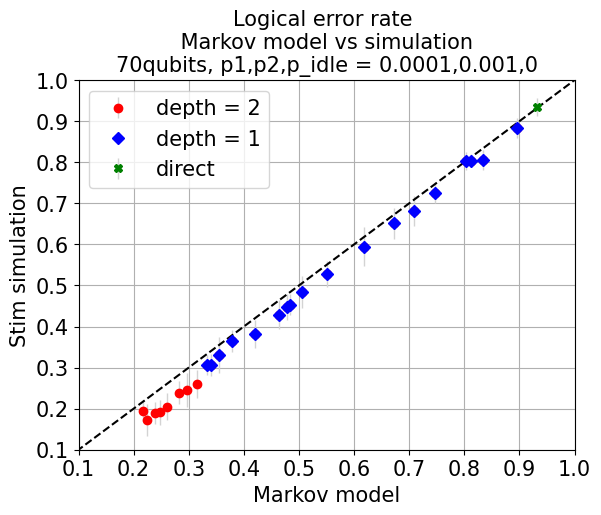}
     \includegraphics[width=0.48\linewidth]{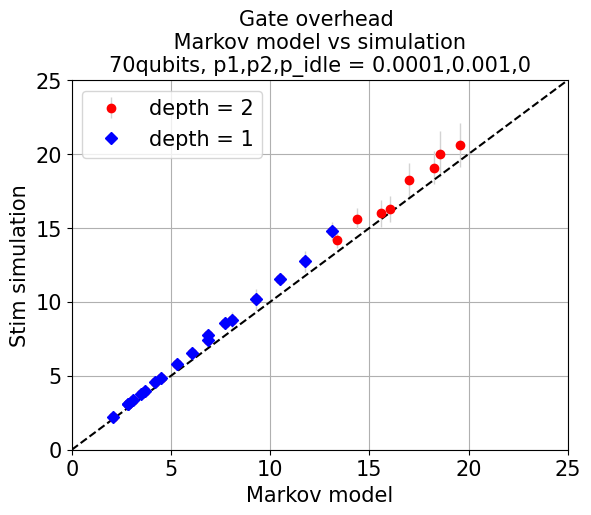}

    \hspace{10pt} (a) \hspace{240pt} (b)

    \vspace{10pt}

 \includegraphics[width=0.48\linewidth]{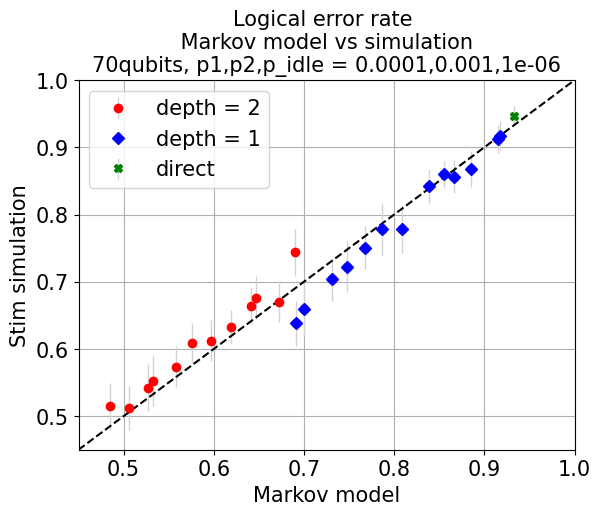}
     \includegraphics[width=0.48\linewidth]{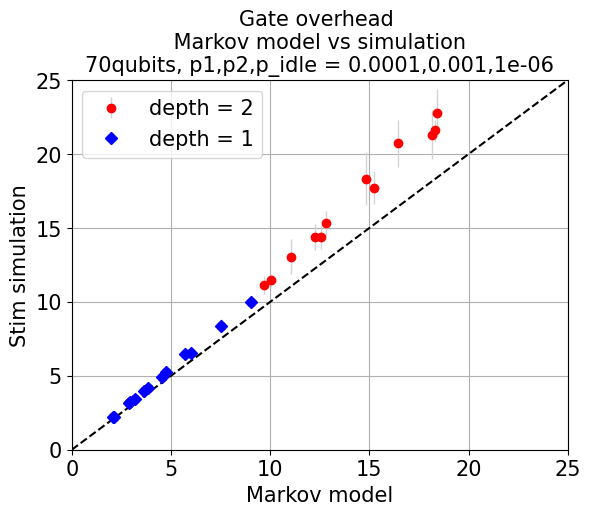}

         \hspace{10pt}   (c) \hspace{240pt} (d)

    \caption{Markov model  compared with the Monte Carlo Stim based simulation at $n=70$ and $p=10^{-3}$. (a) Logical error rate estimates with no idle noise. (b) Gate overhead estimates with no idle noise. (c) Logical error rate estimates with idle noise.  (d) Gate overhead estimates with idle noise. The dashed line is a guide to indicate the same prediction for both models. Values above (below) this line indicate an under- (over-) estimate for the Markov model which is expected to be less accurate. Error bars are standard deviations over 50 random Clifford input circuits (for the Stim simulation). }
    \label{fig:analyt_vs_sim_plog_1}
\end{figure*}

\begin{table}[h!]
\centering
\begin{tabular}{lcccc}
\toprule
 & \multicolumn{2}{c}{$\plog$} & \multicolumn{2}{c}{$\gateoverhead$} \\
\cmidrule(lr){2-3}\cmidrule(lr){4-5}
Tree $(t_\ell,r_\ell)$ & Markov & Stim & Markov & Stim \\
\midrule
direct & 0.9326 & 0.9465 & 1.00 & 1.00 \\
\midrule
(1,\,1) & 0.9150 & 0.9120 & 2.13 & 2.24 \\
(2,\,1) & 0.9176 & 0.9175 & 2.07 & 2.22 \\
(2,\,2) & 0.8548 & 0.8600 & 3.19 & 3.46 \\
(2,\,3) & 0.7863 & 0.7785 & 4.54 & 4.89 \\
(2,\,4) & 0.7311 & 0.7035 & 6.01 & 6.55 \\
(2,\,5) & 0.6996 & 0.6600 & 7.51 & 8.40 \\
(2,\,6) & 0.6908 & 0.6380 & 9.01 & 9.99 \\
(3,\,2) & 0.8667 & 0.8565 & 2.93 & 3.22 \\
(3,\,3) & 0.8086 & 0.7785 & 3.84 & 4.14 \\
(3,\,4) & 0.7675 & 0.7505 & 4.76 & 5.23 \\
(3,\,5) & 0.7480 & 0.7225 & 5.70 & 6.47 \\
(4,\,2) & 0.8852 & 0.8680 & 2.88 & 3.18 \\
(4,\,3) & 0.8391 & 0.8420 & 3.64 & 4.00 \\
\midrule
(1,\,1); (2,\,4) & 0.6726 & 0.6705 & 10.06 & 11.45 \\
(1,\,1); (2,\,5) & 0.6418 & 0.6645 & 12.27 & 14.37 \\
(1,\,1); (3,\,5) & 0.6897 & 0.7435 & 9.67 & 11.13 \\
(1,\,2); (2,\,3) & 0.6194 & 0.6325 & 12.56 & 14.37 \\
(1,\,2); (2,\,4) & 0.5582 & 0.5735 & 15.25 & 17.72 \\
(1,\,2); (2,\,5) & 0.5273 & 0.5420 & 18.18 & 21.31 \\
(1,\,2); (3,\,3) & 0.6471 & 0.6760 & 11.04 & 13.07 \\
(1,\,2); (3,\,4) & 0.5974 & 0.6125 & 12.81 & 15.34 \\
(1,\,2); (3,\,5) & 0.5759 & 0.6085 & 14.84 & 18.34 \\
(1,\,3); (2,\,3) & 0.5058 & 0.5120 & 18.33 & 21.60 \\
(1,\,3); (3,\,3) & 0.5326 & 0.5525 & 16.42 & 20.74 \\
(1,\,3); (3,\,4) & 0.4853 & 0.5145 & 18.39 & 22.79 \\
\bottomrule
\end{tabular}
\caption{Tree parameters used in \cref{fig:pareto}(a) and \cref{fig:analyt_vs_sim_plog_1}(c,d), at
$n=70$, $p=10^{-3}$, with idle noise $10^{-6}$. Each tree is specified by the pair
$(t_{\ell},r_\ell)$ at every level $\ell\ge1$. \emph{Direct} denotes the
direct implementation. Stim values are averaged over 50 random Clifford input circuits with 80
shots each.}
\label{tab:treeparams}
\end{table}

\begin{table}[h!]
\centering
\begin{tabular}{lcccc}
\toprule
 & \multicolumn{2}{c}{$\plog$} & \multicolumn{2}{c}{$\gateoverhead$} \\
\cmidrule(lr){2-3}\cmidrule(lr){4-5}
Tree $(t_\ell,r_\ell)$ & Markov & Stim & Markov & Stim \\
\midrule
direct & 0.9326 & 0.9350 & 1.00 & 1.00 \\
\midrule
(2,\,1) & 0.8961 & 0.8835 & 2.05 & 2.19 \\
(2,\,2) & 0.8040 & 0.8040 & 3.11 & 3.36 \\
(2,\,5) & 0.4837 & 0.4535 & 6.87 & 7.43 \\
(2,\,6) & 0.4198 & 0.3820 & 8.08 & 8.78 \\
(2,\,7) & 0.3793 & 0.3640 & 9.27 & 10.18 \\
(2,\,8) & 0.3550 & 0.3300 & 10.49 & 11.52 \\
(2,\,9) & 0.3409 & 0.3060 & 11.75 & 12.73 \\
(2,\,10) & 0.3330 & 0.3055 & 13.11 & 14.79 \\
(3,\,2) & 0.8122 & 0.8040 & 2.86 & 3.11 \\
(3,\,3) & 0.7089 & 0.6810 & 3.68 & 3.95 \\
(3,\,4) & 0.6176 & 0.5945 & 4.49 & 4.85 \\
(3,\,5) & 0.5502 & 0.5280 & 5.28 & 5.78 \\
(3,\,6) & 0.5060 & 0.4840 & 6.07 & 6.56 \\
(3,\,7) & 0.4790 & 0.4475 & 6.88 & 7.74 \\
(3,\,8) & 0.4632 & 0.4285 & 7.74 & 8.57 \\
(4,\,2) & 0.8332 & 0.8050 & 2.82 & 3.12 \\
(4,\,3) & 0.7460 & 0.7260 & 3.51 & 3.80 \\
(4,\,4) & 0.6728 & 0.6510 & 4.18 & 4.55 \\
\midrule
(1,\,3); (3,\,5) & 0.3155 & 0.2605 & 13.38 & 14.16 \\
(1,\,3); (3,\,6) & 0.2960 & 0.2455 & 14.36 & 15.63 \\
(1,\,4); (3,\,4) & 0.2832 & 0.2390 & 15.56 & 15.99 \\
(1,\,4); (3,\,5) & 0.2599 & 0.2050 & 16.02 & 16.28 \\
(1,\,4); (3,\,6) & 0.2474 & 0.1910 & 17.00 & 18.27 \\
(1,\,5); (3,\,4) & 0.2394 & 0.1905 & 18.28 & 19.08 \\
(1,\,5); (3,\,5) & 0.2250 & 0.1725 & 18.56 & 20.01 \\
(1,\,5); (3,\,6) & 0.2176 & 0.1955 & 19.56 & 20.61 \\
\bottomrule
\end{tabular}
\caption{Tree parameters used in \cref{fig:pareto}(b) and \cref{fig:analyt_vs_sim_plog_1} (a,b), at
$n=70$ and $p=10^{-3}$ with no idle noise. Notation as in \cref{tab:treeparams}.}
\label{tab:treeparams-noidle}
\end{table}

\end{document}